\documentclass[a4paper,envcountsame]{llncs}

\usepackage{enumerate,amsfonts,graphicx,url,stfloats,subfig}
\usepackage{amssymb,amsmath,color,verbatim,bm}
\usepackage{stmaryrd}
\usepackage{paralist}
\usepackage{environ,xspace}
\usepackage{thmtools,thm-restate,wrapfig}

\usepackage{tikz}
\usetikzlibrary{arrows,automata,shapes}

\usepackage[color=light gray]{todonotes}

\DeclareMathAlphabet{\mathpzc}{OT1}{pzc}{m}{it}

\newcommand{\figref}[1]{Fig.~\ref{fig:#1}}

\newcommand{\ignore}[1]{}

\newcommand{\s}{s}   
\newcommand{\str}{t}
\newcommand{\eos}{\texttt{\#}}
\newcommand{\pad}{\ensuremath{\sqcup}}
\newcommand{\sos}{\emptystring}
\newcommand{\seos}[1]{{#1}^{\pad,\eos}}
\newcommand{\reset}{\rho}

\newcommand{\true}{true}

\newcommand{\Ts}{w}   
\newcommand{\Tstr}{v}

\newcommand{\hd}{d_{M}}

\newcommand{\lone}[1]{||{#1}||_1}

\newcommand{\diff}{\mathtt{diff}}
\newcommand{\dist}{d}

\newcommand{\out}{w'} 

\newcommand{\Rat}{\mathbb{Q}}
\newcommand{\Reals}{\mathbb{R}}

\newcommand{\Ratp}{\mathbb{Q}^+}

\newcommand{\Realsp}{\mathbb{R}^+}
\newcommand{\Realsinf}{\mathbb{\Reals}^\infty}

\newcommand{\funcDefinedBy}[1]{\ensuremath{\llbracket #1\rrbracket}}
\newcommand{\eps}{\epsilon}
\newcommand{\bound}{B}
\newcommand{\dur}{\lambda}

\newcommand{\tran}{\ensuremath{\mathpzc{T}}}
\newcommand{\trans}{\ensuremath{\mathpzc{T}}}

\newcommand{\R}{R}

\newcommand{\emptystring}{\eps}
\newcommand{\es}{\phi}

\newcommand{\ftime}[1]{{#1}_{C}}
\newcommand{\timed}[1]{\textrm{timed}({#1})}

\newcommand{\wiggle}{\lambda}
\newcommand{\Wiggles}{\Lambda}

\newcommand{\dtM}{d_{TM}}

\newcommand{\dad}{d_{AD}}
\newcommand{\dS}{d_{S}}
\newcommand{\timedom}{I}
\newcommand{\twa}{weighted timed automaton}
\newcommand{\wta}{WTA}
\newcommand{\symb}{s}

\newcommand{\tim}{\delta}

\tikzstyle{smalltext}=[font=\fontsize{7}{8}\selectfont]
\tikzstyle{captiontext}=[font=\fontsize{8}{8}\selectfont]
\tikzstyle{state}=[draw, ellipse, minimum height=5mm,
                   minimum width=5mm, inner sep=1pt,
                   text=black,
                   font=\fontsize{8}{8}\selectfont,
                   semithick]

\newcommand{\rej}{\text{rej}}

\newcommand{\fst}{{\sc fst}\xspace}

\newcommand{\nlogspace}{\textsc{NLogspace}}
\newcommand{\pspace}{\textsc{PSpace}}

\newcommand{\lpair}[2]{\langle {#1}, {#2} \rangle}

\newcommand{\alphI}{\Sigma}
\newcommand{\alphO}{\Gamma}

\newcommand{\baut}{\bar{\aut}}
\newcommand{\blaut}{\baut^L}
\newcommand{\braut}{\baut^R}
\newcommand{\bautI}{\baut_\alphI}
\newcommand{\bautO}{\baut_\alphO}

\newcommand{\dI}{d_{\alphI}}
\newcommand{\dO}{d_{\alphO}}

\newcommand{\diffI}{\diff_{\alphI}}
\newcommand{\diffO}{\diff_{\alphO}}

\newcommand{\dom}[1]{\mathrm{dom}(#1)}

\newcommand{\N}{{\cal N}}

\newcommand{\aut}{{\cal A}}
\newcommand{\autI}{{\cal A}_{\alphI}}
\newcommand{\autO}{{\cal A}_{\alphO}}
\newcommand{\Q}{{\mathbb{Q}}}

\newcommand{\lang}{{\cal L}}

\newcommand{\Acc}{\mathsf{Acc}}

\newcommand{\fsum}{\textsc{Sum}}

\newcommand{\timedL}[1]{{\cal TL}({#1})}
\newcommand{\untimed}[1]{\textrm{untimed}(#1)}

\newcommand{\transition}[1]{\rightarrow^{#1}}

\newcommand{\transitionComb}[2]{\rightarrow^{#1}_{#2}}

\newcommand{\traces}[1]{\boldsymbol{[}{#1}\boldsymbol{]}}

\newcommand{\asc}{{ASC\xspace}}
\newcommand{\ascs}{{ASCs\xspace}}
\newcommand{\coc}{{CC}}

\newcommand{\ckt}{{\cal C}}
\newcommand{\Inp}{{\cal I}}
\newcommand{\Out}{{\cal O}}
\newcommand{\Y}{{\cal Y}}
\newcommand{\Z}{{\cal Z}}

\newcommand{\ltuple}[2]{\{ {#1}^1, \ldots, {#1}^{#2} \}}

\newcommand{\genericAscDiagram}{
\begin{tikzpicture}

\node[rectangle,draw,minimum height=3.3cm,minimum width=2.4cm] at (0.2,0.3) {combinatorial};
\node[] at (0.2,-0.3) {circuit};

\begin{scope}
\draw[->] (-1.8,1.8) to node[above] {$i^1$} (-1,1.8);
\node[] at (-1.5, 1.62) {$\vdots$};
\draw[->] (-1.8,0.8) to node[above] {$i^m$} (-1,0.8);
\end{scope}

\begin{scope}[xshift=0.4cm]
\draw[->] (1,1.8) to node[above] {$o^1$} (1.8,1.8);
\node[] at (1.3, 1.62) {$\vdots$};
\draw[->] (1.0,0.8) to node[above] {$o^n$} (1.8,0.8);
\end{scope}

\begin{scope}[yshift=-2cm]
\draw[->] (-2.0,1.8) to node[above] {$y^1$} (-1,1.8);
\node[] at (-1.5, 1.62) {$\vdots$};
\draw[->] (-1.8,0.8) to node[above] {$y^k$} (-1,0.8);

\draw[-] (-1.8,0.8) to (-1.8,0.3);
\draw[-] (-1.8,0.3) to (-0.3,0.3);
\node[rectangle,draw,minimum height=0.5cm,minimum width=0.6cm] at (0,0.3) {$d^k$};
\draw[->] (2.2,0.3) to (0.3,0.3);
\draw[-] (2.2,0.8) to (2.2,0.3);
\draw[-] (2.2,0.8) to node[above] {$z^k$} (1.4,0.8);

\node[] at (0, -0.1) {$\vdots$};

\draw[-] (-2,1.8) to (-2,-0.8);
\draw[-] (-2,-0.8) to (-0.3,-0.8);
\node[rectangle,draw,minimum height=0.5cm,minimum width=0.6cm] at (0,-0.8) {$d^1$};
\draw[->] (2.4,-0.8) to (0.3,-0.8);
\draw[-] (2.4,1.8) to (2.4,-0.8);
\draw[-] (2.4,1.8) to node[above] {$z^1$} (1.4,1.8);

\node[] at (1.7, 1.53) {$\vdots$};
\end{scope}
\end{tikzpicture}
}

\newcommand{\oscillatorDiagram}
{
\begin{tikzpicture}

\node[circle,draw,minimum size=0.5cm] (V) at (0,-0.5) {$\vee$};
\node[rectangle,draw,minimum height=0.5cm,minimum width=0.5cm] (D) at (0,-1.7) {delay $1$};
\node[]  at (0,-2.3) {Circuit $\ckt$.};

\draw[-] (-1.3,0.3) to node[above] {$i$} (-0.8,0.3);
\draw[->] (-0.8,0.3) to (V);
\draw[->] (-1,-1.5) to node[above]{$y$} (V);
\draw[-] (-1,-1.5) to (-1,-1.7);
\draw[-] (-1,-1.7) to (D.west);

\draw[-] (V) to (1.1,-0.5);
\draw[-] (1.1,-0.5) to  (1.1,-1.7);
\node[circle,fill,inner sep = 0pt,
      minimum size = 2pt] (A) at (1.1, -0.5) {};
\draw[->] (1.1,-1.7) to node [above] {$z$}  (D.east);
\draw[->] (1.1,-0.5) to node[above]{$o$} (1.5,-0.5);

\begin{scope}[xshift=3.3cm]
\node[circle,draw,minimum size=0.5cm] (V) at (0,-0.5) {$\vee$};
\node[rectangle,draw,minimum height=0.5cm,minimum width=0.5cm] (D) at (0,-1.7) {delay $1$};
\draw[-] (-1.3,0.3) to node[above] {$i$} (-0.8,0.3);
\draw[->] (-0.8,0.3) to (V);
\draw[->] (-1,-1.5) to node[above]{$y$} (V);
\draw[-] (-1,-1.5) to (-1,-1.7);
\draw[-] (-1,-1.7) to (D.west);

\draw[->] (V) to  (1.6,-0.5);
\node[] at (1.4,-0.3) {$o$};
\draw[-] (1.1,0.3) to  (1.1,-1.7);
\draw[-] (1.1,0.3) to  (-0.8,0.3);
\draw[->] (1.1,-1.7) to node [above] {$z$}  (D.east);
\node[]  at (0,-2.3) {Circuit $\ckt'$.};
\node at (0,4.5) {};
\end{scope}

\end{tikzpicture}
}

\newcommand{\tint}[3]{
\node[circle,fill,inner sep = 0pt,
      minimum size = 2pt] (A) at (#2, #1) {};
\node[circle,draw,inner sep = 0pt,
      minimum size = 2pt] (B) at (#3, #1) {};
\draw (#2,#1) -- (#3-0.04,#1);
}

\newcommand{\joint}[3]{
\draw[densely dotted] (#3,#1) -- (#3,#2);
}

\newcommand{\oscillatorGraph}
{
\begin{tikzpicture}
\newcommand{\spread}{0.8}
\newcommand{\perturb}{0.3}
\draw[->] (0,0) -- (0,3.5);
\draw (0,0) -- (0,-0.3);

\draw[->] (0,0) -- (6,0);
\draw (0,0) -- (-0.3,0);

\node at (-0.2,0.5) {$o_2'$};
\node at (-0.2,1.5) {$o_2$};
\node at (-0.2,2.5) {$i_2$};

\draw (-0.05,1) -- (0.05,1);
\draw (-0.05,2) -- (0.05,2);
\draw (-0.05,3) -- (0.05,3);

\draw[dashed] (0.05,1) -- (6.05,1);
\draw[dashed] (0.05,2) -- (6.05,2);
\draw[dashed] (0.05,3) -- (6.05,3);

\foreach \x in {1,2,3,4,5,6}
\draw (\x*\spread,-0.05) -- node[below] {$\x$} (\x*\spread,0.05);

\tint{2.6}{0}{\perturb}
\tint{2.3}{\perturb}{6}
\joint{2.6}{2.3}{\perturb}

\foreach \x in {0,1,2,3,4,5}
{
\joint{1.6}{1.3}{\x*\spread}
\tint{1.6}{\x*\spread}{\x*\spread +\perturb}
\ifthenelse{\x=5}{\tint{1.3}{\x*\spread+ \perturb}{6}}{\tint{1.3}{\x*\spread+ \perturb}{\x*\spread+\spread}}
\joint{1.6}{1.3}{\x*\spread+\perturb}
}

\tint{0.6}{0}{\perturb}
\tint{0.3}{\perturb}{\spread}
\joint{0.6}{0.3}{\perturb}

\joint{0.6}{0.3}{\spread}

\tint{0.6}{\spread}{\perturb+\spread}
\tint{0.3}{\perturb+\spread}{6}
\joint{0.6}{0.3}{\perturb+\spread}

\end{tikzpicture}
}

\newcommand{\cadlagGraph}{
\begin{tikzpicture}
\draw[->] (0,0) -- (0,3);
\draw (0,0) -- (0,-0.3);

\draw[->] (0,0) -- (6,0);
\draw (0,0) -- (-0.3,0);

\node at (-0.2,0.75) {a};
\node at (-0.2,1.5) {b};
\node at (-0.2,2.25) {c};

\draw (-0.05,0.75) -- (0.05,0.75);
\draw (-0.05,1.5) -- (0.05,1.5);
\draw (-0.05,2.25) -- (0.05,2.25);

\node at (1.3,-0.3) {1.3};
\node at (2,-0.3) {2};
\node at (3.7,-0.3) {3.7};
\node at (5,-0.3) {5};

\draw (1.3,-0.05) -- (1.3,0.05);
\draw (2,-0.05) -- (2,0.05);
\draw (2.9,-0.05) -- (2.9,0.05);
\draw (3.7,-0.05) -- (3.7,0.05);
\draw (5,-0.05) -- (5,0.05);

\tint{0.75}{0}{1.3}
\tint{1.5}{1.3}{2}
\tint{0.75}{2}{3.7}
\tint{2.25}{3.7}{5}
\node[circle,fill,inner sep = 0pt,
      minimum size = 2pt] (A) at (2.9, 0.75) {};
\node[circle,fill,inner sep = 0pt,
      minimum size = 2pt] (A) at (5, 0.75) {};
\end{tikzpicture}
}

\newcommand{\littledots}[2]{
\node[] at (#1,#2) {.};
\node[] at (#1,#2+0.11) {.};
\node[] at (#1,#2-0.11) {.};
}

\newcommand{\succintReachabilityCircuit}{
\begin{tikzpicture}
\newcommand{\ax}{3}
\newcommand{\bx}{2}
\newcommand{\by}{-1}
\newcommand{\cy}{\by-1.0}

\littledots{0.4}{-0.2}
\littledots{\ax+1.5}{-0.2}
\littledots{\bx+1.5}{\cy-0.25}
\littledots{\bx-0.4}{\cy-0.25}

\node at (0.1,-0.2) {$\vec{w}$};
\node at (\ax+2.2,-0.2) {$\vec{u}$};
\node at (\bx+0.8,\cy-0.25) {$\vec{u}$};
\node at (\bx-0.8,\cy-0.25) {$\vec{v}$};

\node[circle,draw,minimum size=0.4cm,inner sep=1pt] (AndOne)  at (\ax+0.9, 0) {$\wedge$}; 
\node[circle,draw,minimum size=0.4cm,inner sep=1pt] (AndTwo) at (\ax+0.25, -0.4) {$\wedge$}; 
\draw[->] (0,0) to (AndOne.west);
\draw[->] (0,-0.4) to (AndTwo.west);

\draw[->] (AndOne.east) to (\ax+2.5,0);
\draw[->] (AndTwo.east) to (\ax+2.5,-0.4);

\node[rectangle,draw,minimum height=0.7cm,minimum width=1.2cm] (Equality) at (\ax+3.1,-0.2) {$\vec{u} = \vec{t}$};

\draw[-] (\bx-1.2,0) to (\bx-1.2, \by+0.15);
\draw[->] (\bx-1.2, \by+0.15) to (\bx-0.7, \by+0.15);

\draw[-] (\bx-1.0,-0.4) to (\bx-1.0, \by+0.25);
\draw[->] (\bx-1.0, \by+0.25) to (\bx-0.7, \by+0.25);

\node[rectangle,draw,minimum height=0.3cm,minimum width=0.3cm] (DelOne) at (\bx,\cy-0.5) {$1$};
\draw[->] (\bx-1.2, \by-0.15) to (\bx-0.7, \by-0.15);
\draw[-] (\bx-1.2,\cy-0.5) to (\bx-1.2, \by-0.15);
\draw[-] (\bx-1.2,\cy-0.5) to (DelOne.west);
\draw[<-] (DelOne.east) to (\ax+1.9, \cy-0.5);
\draw[-] (\ax+1.9, \cy-0.5) to (\ax+1.9, 0);

\node[rectangle,draw,minimum height=0.3cm,minimum width=0.3cm] (DelTwo) at (\bx,\cy) {$1$};
\draw[->] (\bx-1.0, \by-0.25) to (\bx-0.7, \by-0.25);
\draw[-] (\bx-1.0,\cy) to (\bx-1.0, \by-0.25);
\draw[-] (\bx-1.0,\cy) to (DelTwo.west);
\draw[<-] (DelTwo.east) to (\ax+1.7, \cy);
\draw[-] (\ax+1.7, \cy) to (\ax+1.7, -0.4);

\node[rectangle,draw,minimum height=0.8cm,minimum width=1.4cm] at (\bx,\by) {$E(\vec{v},\vec{w})$};

\draw[-] (\bx+0.7,\by) to (\ax+0.9,\by);
\draw[->] (\ax+0.25,\by) to (\ax+0.25, -0.6);
\draw[->] (\ax+0.9,\by) to (\ax+0.9, -0.2);

\node[rectangle,draw,minimum height=0.7cm,minimum width=1.2cm] (Oscillator) at (\ax+3.9,-1.5) {Oscillator};
\draw[-] (Equality.east) to (\ax+3.9,-0.2);
\draw[->] (\ax+3.9,-0.2) to (Oscillator.north);

\draw[->] (Oscillator.east) to (\ax+5.3,-1.5);
\node at (\ax+5.0,-1.3) {$o$};

\end{tikzpicture}
}

\newcommand{\peak}[3]
{
\draw[-, thick] (#1,#2) to (#1,#3);
\node[circle,fill,inner sep = 0pt,
      minimum size = 2pt] (A) at (#1, #3) {};
}

\newcommand{\OfflineOnlineExampleGraph}
{
\begin{tikzpicture}

\newcommand{\spread}{1.0}
\newcommand{\perturb}{0.15}

\draw[->] (0,0) -- (0,4.1);
\draw (0,0) -- (0,-0.3);

\draw[->] (0,0) -- (5.0,0);
\draw (0,0) -- (-0.3,0);

\node at (-0.5,3.5) {$i_1$};
\node at (-0.5,2.5) {$i_2$};
\node at (-0.5,1.5) {$o_1$};
\node at (-0.5,0.5) {$o_2$};

\node at (-0.2,3.6) {\tiny $r$};
\node at (-0.2,3.3) {\tiny $\emptyset$};
\draw[-] (-0.05,3.6) to (0.05,3.6);
\draw[-] (-0.05,3.3) to (0.05,3.3);

\node at (-0.2,2.6) {\tiny $r$};
\node at (-0.2,2.3) {\tiny $\emptyset$};
\draw[-] (-0.05,2.6) to (0.05,2.6);
\draw[-] (-0.05,2.3) to (0.05,2.3);

\node at (-0.2,1.6) {\tiny $\top$};
\node at (-0.2,1.3) {\tiny $\bot$};
\draw[-] (-0.05,1.6) to (0.05,1.6);
\draw[-] (-0.05,1.3) to (0.05,1.3);

\node at (-0.2,0.6) {\tiny $\top$};
\node at (-0.2,0.3) {\tiny $\bot$};
\draw[-] (-0.05,0.6) to (0.05,0.6);
\draw[-] (-0.05,0.3) to (0.05,0.3);

\draw (-0.05,1) -- (0.05,1);
\draw (-0.05,2) -- (0.05,2);
\draw (-0.05,3) -- (0.05,3);
\draw (-0.05,4) -- (0.05,4);

\draw[dashed] (0.05,1) -- (4.85,1);
\draw[dashed] (0.05,2) -- (4.85,2);
\draw[dashed] (0.05,3) -- (4.85,3);

\draw[-] (0,3.3) to (\spread,3.3);
\peak{0}{3.3}{3.6}

\draw[-] (0,2.3) to (\spread+0.1,2.3);
\peak{0}{2.3}{2.6}

\foreach \x in {1,2,3,4}
{

\draw[-] (\x*\spread-\spread,3.3) to (\x*\spread,3.3);
\peak{\x*\spread}{3.3}{3.6}

\draw[-] (\x*\spread-\spread,2.3) to (\x*\spread+\perturb,2.3);
\draw (\x*\spread,-0.05) -- node[below] {$\x$} (\x*\spread,0.05);
\draw[dotted,-] (\x*\spread,0) to (\x*\spread,4);
\peak{\x*\spread+\x*\perturb}{2.3}{2.6}

\node[circle,fill,inner sep = 0pt,
      minimum size = 2pt] (A) at (\x*\spread, 1.6) {};

\node[circle,fill,inner sep = 0pt,
      minimum size = 2pt] (A) at (\x*\spread+\x*\perturb, 0.3) {};

}

\draw[-] (4*\spread,3.3) to (5*\spread,3.3);
\draw[-] (4*\spread,2.3) to (5*\spread,2.3);

\draw[-] (\spread,1.6) to (5*\spread,1.6);
\draw[-] (\spread+\perturb,0.3) to (5*\spread,0.3);

\begin{scope}[xshift=6.5cm]
\draw[->] (0,0) -- (0,4.1);
\draw (0,0) -- (0,-0.3);

\draw[->] (0,0) -- (5.0,0);
\draw (0,0) -- (-0.3,0);

\node at (-0.5,3.5) {$i_1$};
\node at (-0.5,2.5) {$i_2$};
\node at (-0.5,1.5) {$o_1'$};
\node at (-0.5,0.5) {$o_2'$};

\node at (-0.2,3.6) {\tiny $r$};
\node at (-0.2,3.3) {\tiny $\emptyset$};
\draw[-] (-0.05,3.6) to (0.05,3.6);
\draw[-] (-0.05,3.3) to (0.05,3.3);

\node at (-0.2,2.6) {\tiny $r$};
\node at (-0.2,2.3) {\tiny $\emptyset$};
\draw[-] (-0.05,2.6) to (0.05,2.6);
\draw[-] (-0.05,2.3) to (0.05,2.3);

\node at (-0.2,1.6) {\tiny $\top$};
\node at (-0.2,1.3) {\tiny $\bot$};
\draw[-] (-0.05,1.6) to (0.05,1.6);
\draw[-] (-0.05,1.3) to (0.05,1.3);

\node at (-0.2,0.6) {\tiny $\top$};
\node at (-0.2,0.3) {\tiny $\bot$};
\draw[-] (-0.05,0.6) to (0.05,0.6);
\draw[-] (-0.05,0.3) to (0.05,0.3);

\draw (-0.05,1) -- (0.05,1);
\draw (-0.05,2) -- (0.05,2);
\draw (-0.05,3) -- (0.05,3);
\draw (-0.05,4) -- (0.05,4);

\draw[dashed] (0.05,1) -- (4.85,1);
\draw[dashed] (0.05,2) -- (4.85,2);
\draw[dashed] (0.05,3) -- (4.85,3);

\draw[-] (0,3.3) to (\spread,3.3);
\peak{0}{3.3}{3.6}

\draw[-] (0,2.3) to (\spread+0.1,2.3);
\peak{0}{2.3}{2.6}

\node[circle,fill,inner sep = 0pt,
      minimum size = 2pt] (A) at (0, 1.6) {};

\node[circle,fill,inner sep = 0pt,
      minimum size = 2pt] (A) at (0, 0.6) {};

\foreach \x in {1,2,3,4}
{
\draw (\x*\spread,-0.05) -- node[below] {$\x$} (\x*\spread,0.05);
\draw[dotted,-] (\x*\spread,0) to (\x*\spread,4);

\draw[-] (\x*\spread-\spread,3.3) to (\x*\spread,3.3);
\peak{\x*\spread}{3.3}{3.6}

\draw[-] (\x*\spread-\spread,2.3) to (\x*\spread+0.1,2.3);
\peak{\x*\spread+\x*\perturb}{2.3}{2.6}

\node[circle,fill,inner sep = 0pt,
      minimum size = 2pt] (A) at (\x*\spread, 1.6) {};

\draw[-] (\x*\spread+\x*\perturb-\perturb-\spread,0.6) to (\x*\spread+\x*\perturb-\perturb,0.6);
\draw[-] (\x*\spread+\x*\perturb-\perturb,0.3) to (\x*\spread+\x*\perturb,0.3);
}

\draw[-] (4*\spread,3.3) to (5*\spread,3.3);
\draw[-] (4*\spread,2.3) to (5*\spread,2.3);

\draw[-] (0,1.6) to (5*\spread,1.6);
\draw[-] (4*\spread+4*\perturb,0.6) to (4.5*\spread+4*\perturb,0.6);
\end{scope}

\end{tikzpicture}
}

\title{Lipschitz Robustness of Timed I/O Systems}
\author{Thomas A. Henzinger, Jan Otop and  Roopsha Samanta}
\institute{IST Austria \\
\email{\{tah,jotop,rsamanta\}@ist.ac.at}}

\begin{document}

\maketitle

\begin{abstract}
We present the first study of robustness of systems that are both timed as well as reactive (I/O). 
We study the behavior of such timed I/O systems in the presence of {\em uncertain inputs} 
and formalize their robustness using the analytic notion of Lipschitz continuity. 
Thus, a timed I/O system is $K$-(Lipschitz) 
robust if the perturbation in its output is at most $K$ times the perturbation in its input. 
We quantify input and output perturbation using {\em similarity functions} 
over timed words such as the timed version of the Manhattan distance
 and the Skorokhod distance.
We consider two models of timed I/O systems --- timed transducers 
and asynchronous sequential circuits. 
 While $K$-robustness is undecidable even for discrete transducers, 
we identify a class of timed transducers which admits a polynomial space
decision procedure for $K$-robustness. 
For asynchronous sequential circuits, we reduce $K$-robustness w.r.t. timed Manhattan distances 
to $K$-robusness of discrete letter-to-letter transducers and show 
\pspace-compeleteness of the problem.


\end{abstract}

\section{Introduction}
\label{s:intro}
\setlength{\floatsep}{-80pt}

Real-time systems operating in physical environments are increasingly
commonplace today.
An inherent problem faced by such computational systems is {\em
input uncertainty} caused by sensor inaccuracies, imprecise environment assumptions etc. 
This means that the input data may be noisy and/or may have timing errors.
Hence, it is not enough for such a timed I/O system to be functionally correct.  
It is also desirable that the system be {\em continuous} or {\em robust}, i.e., the system
behavior degrade smoothly in the presence of input disturbances
\cite{Henzinger08}. We illustrate this property with two examples of timed I/O systems. 

\begin{example}
\label{ex:tick}
Consider two timed I/O systems which process a sequence of ticks and 
calibrate the intervals between the ticks (see Fig. \ref{f:ticktime}). In particular, the 
goal is to track if an interval is greater than some given $\Delta$.
The first timed I/O system $\trans$ is an {\em offline} processor; 
upon arrival of each request, $\trans$ 
waits till the next request, and outputs $\top$ if 
the interval is less than or equal to $\Delta$ and $\bot$ otherwise. 
The second timed I/O system $\trans'$ is an {\em online} processor;
$\trans'$ starts generating $\top$ immediately upon arrival of each request, 
and switches its output to $\bot$ after $\Delta$ time, until the arrival of the next request.

Consider two periodic tick sequences: $i_1$ and $i_2$ as shown in Fig. \ref{f:ticktime}.
The duration between ticks in $i_1$, $i_2$ is $\Delta$, $\Delta+\eps$, respectively. 
Thus $i_2$ can be viewed as a timing distortion of $i_1$. 
While the output $o_1$ of $\trans$ on $i_1$ is a constant sequence of $\top$,
the output $o_2$ of $\trans$ on $i_2$ consists of $\bot$ entirely. Thus, a small timing 
perturbation in the input of $\trans$ can cause a large perturbation in its output. 
On the other hand, a small timing perturbation in the input of $\trans'$ only causes a 
proportionally small perturbation in its output. Indeed, while the output 
$o_1'$ of $\trans'$ on $i_1$ is also a constant sequence of $\top$, 
the output $o_2'$ of $\trans'$ on $i_2$ is a sequence of $\top$, with periodic $\bot$ intervals of $\eps$-duration.  
\end{example}

\begin{figure}
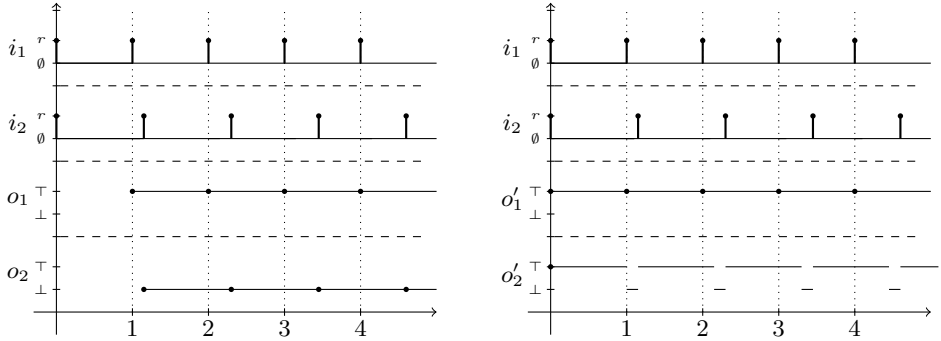

\begin{center}
\OfflineOnlineExampleGraph
\end{center}
\caption{Timing distortion}
\label{f:ticktime}
\end{figure}

\begin{example}
\label{ex:oscillator}
Consider two asynchronous sequential circuits $\ckt$ and $\ckt'$ shown in Fig. \ref{f:osc}.  
For each circuit, the input is $i$, the output is $i \vee y$ and the value of variable $y$ at time $t$ equals the value of variable $z$ at time $t-1$. 
In circuit $\ckt$, variable $z$ equals $i \vee y$ and in circuit $\ckt'$, variable $z$ equals $i$. Initially $y$ is set to $0$.
\begin{figure}[t!h]
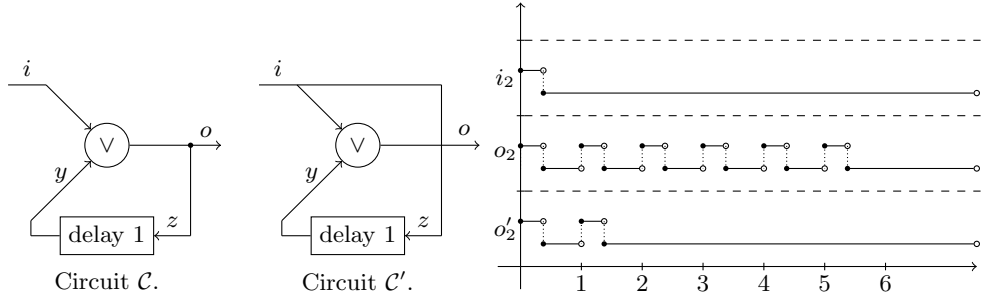

\oscillatorDiagram
\bigskip
\oscillatorGraph
\caption{Transient fault}
\label{f:osc}
\end{figure}

Consider inputs $i_1$ and $i_2$ such that $i_1$ is constantly $0$,  
and $i_2$ is $1$ in the interval $[0,\eps)$ and $0$ otherwise (see Fig. \ref{f:osc}).
Thus, $i_2$ can be viewed as representing a transient fault in $i_1$. 
The outputs of both $\ckt$ and $\ckt'$ for $i_1$ are constantly $0$.
For $i_2$, $\ckt$ produces a periodic sequence that equals $1$ exactly in the intervals 
$[0,\eps), [1,1+\eps), [2,2+\eps) \ldots$, whereas $\ckt'$ produces the output that equals $1$ only in the intervals
$[0,\eps)$ and $[1,1+\eps]$. Thus, the effect of a small input perturbation propagates forever 
in the output of $\ckt$. On the other hand, the effect of a small input perturbation is limited to a bounded time 
in the output of $\ckt'$. 

\end{example}

We present the first study of robustness of systems 
that are both timed as well as reactive (I/O).  
We formalize robustness of timed I/O systems as Lipschitz
continuity~\cite{foundation,DBLP:journals/corr/HenzingerOS14}. A function is Lipschitz-continuous if its output changes
proportionally to every change in the input.  Given a constant $K$ and
{\em similarity functions} $\dI$, $\dO$ for computing the input, output
perturbation, respectively, a timed I/O system $\tran$ is defined to be
$K$-Lipschitz robust (or simply, $K$-robust) w.r.t. $\dI$, $\dO$ if for all
timed words $\Ts,\Tstr$ in the domain of $\tran$ with finite $\dI(\Ts,\Tstr)$, 
$\dO(\tran(\Ts),\tran(\Tstr)) \le K \dI(\Ts,\Tstr)$. 

In this work, we focus on $K$-robustness of two models of timed I/O systems ---
timed transducers (Ex. \ref{ex:tick}) and asynchronous sequential circuits (\ascs) (Ex. \ref{ex:oscillator}). 
We define a timed transducer as a timed automaton over an alphabet 
partitioned into an input alphabet $\dI$ and an output alphabet $\dO$.  
A timed transducer defines a transduction over timed words, or a 
{\em timed relation}. 
An \asc~is composed of a combinational circuit (CC),  
{\em delay elements} and {\em feedback loops} (see, for instance, Fig. \ref{f:osc}).
An \asc~also defines a timed relation. 
However, timed transducers and \ascs~are are expressively incomparable. 
A simple \asc~that delays its inputs by $1$ time unit is not expressible by timed transducers ---
intuitively, the timed transducer at time $1$ would need to remember arbitrarily many timed events from the interval $[0,1)$.
Conversely, a simple timed transducer that outputs $1$ if the duration between preceding input events is greater than $1$, and $0$
otherwise cannot be expressed by any~\asc.

Since $K$-robustness is undecidable for discrete
transducers \cite{DBLP:journals/corr/HenzingerOS14}, it is also undecidable in
general for our timed transducers. 
We identify a class of timed transducers, called {\em timed-synchronized transducers}, 
which admit decidable $K$-robustness. 
This class includes {\em timed Mealy machines}, i.e., 
timed transducers that accept timed words with alternating input and output letters. 
The key idea behind decidability is a reduction of $K$-robustness of timed-synchronized transducers 
to emptiness of weighted timed automata, given similarity functions computable by weighted timed automata. 
In particular, our results for timed-synchronized transducers include the following:
\begin{compactenum}
\item $K$-robustness 
is \pspace-complete for timed Manhattan distances.
\item $K$-robustness 
is \pspace-complete for accumulated delay distances,  
under practically-viable environment assumptions (e.g., minimum symbol persistence).
\item $K$-robustness 
is \pspace-complete if the input perturbation is computed 
as a Skorokhod distance and the output perturbation is computed 
as a timed Manhattan distance. 
\end{compactenum}

We reduce $K$-robustness of \ascs~w.r.t. timed
Manhattan distances to $K$-robustness of discrete letter-to-letter transducers, 
and show that $K$-robustness of \ascs~is \pspace-complete. The reduction consists 
of two steps. 
First, we show that on inputs that are step functions, \ascs~behave like discrete letter-to-letter transducers. 
Second, we show that if an \asc~is not $K$-robust w.r.t. timed Manhattan distances, 
there exists a witness consisting of a pair of inputs that are step functions.



The paper is organized as follows. We first recall necessary formalisms~(Sec. \ref{s:preliminaries}) 
and present our models of timed I/O systems~(Sec. \ref{s:formalModels}). We formalize our notion of robustness 
for such systems~(Sec. \ref{s:problem}) and define the similarity functions of interest~(Sec. \ref{s:formalDistances}). 
We then present our results on robustness analysis of timed transducers~(Sec. \ref{s:robustnessAutomata})
and \ascs~(Sec. \ref{s:robustnessCircuits}) w.r.t. various similarity functions. \\

\noindent {\bf Related work}.
Robustness of systems has been studied in different contexts such as robust
control \cite{ZDG96}, timed automata \cite{GHJ97}, discrete transducers
\cite{foundation,DBLP:journals/corr/HenzingerOS14} and sequential circuits~\cite{DHLN10}.
However, none of these results are directly applicable to robustness 
of timed I/O systems. There are two main reasons. First, we are interested in 
robustness w.r.t. {\em input} perturbation. Second, 
timed I/O systems exhibit both discrete and continuous behavior.
Robust control typically involves reasoning about continuous
state-spaces and focuses on designing controllers that function properly in the
presence of perturbation in various internal parameters of a system's model.
The study of robustness of timed automata 
focuses on the design of models whose language is robust to infinitesimal timing 
perturbation (e.g. clock drifts). This work does not explicitly consider input perturbation, 
nor does it focus on quantifying the effect of input perturbation on the output. 
Robustness analysis of finite-state transducers is limited to 
purely discrete systems and data.  
In \cite{DHLN10}, the authors study the robustness of synchronous sequential circuits 
modeled as discrete Mealy machines.  Their notion
of robustness bounds the persistence of the effect of a sporadic disturbance and is also 
limited to discrete data.

In other related work~\cite{MRT13,CHR10,BGHJ09}, the authors develop
different notions of robustness for reactive systems, with
$\omega$-regular specifications, interacting with uncertain
environments. There has also been foundational work on continuity and robustness
analysis of software programs manipulating numbers \cite{RTSS09,CGL10,CGLN11}.  
%

\section{Preliminaries}
\label{s:preliminaries}
\newcommand{\op}{\ensuremath{\otimes}}
\newcommand{\buchi}{B\"uchi}


\subsection{Timed automata}

We briefly present basic notions regarding timed automata. We refer the reader to 
\cite{DBLP:conf/sfm/AlurM04} for a comprehensive survey on timed automata.

\noindent{\bf Timed words.} Let $\Realsp$, $\Ratp$ denote the set of all {\em nonnegative} real numbers, rational numbers, respectively. 
A (finite or infinite) \emph{timed word} over an alphabet $\Sigma$ is a word over $(\Sigma, \Realsp)$: 
$(a_0, t_0)(a_1, t_1) \ldots$ such that $t_0, t_1, \ldots$ is a weakly increasing sequence. 
A pair $(a,t)$ is referred to as \emph{an event}. We denote by 
$\timedL{\Sigma}$ the set of all timed words over $\Sigma$.
For a timed word $w = (a_0, t_0) (a_1, t_1) \ldots$ we define 
$\untimed{w} = a_0 a_1 \ldots$ as the projection of $w$ on the $\Sigma$ component.


\noindent{\bf Disjoint union of timed words.} Let $w_1, w_2$ be timed words over the alphabet  $\Sigma$.
We define \emph{the disjoint union} of $w_1$ and $w_2$, denoted $w_1 \oplus w_2$,
as the union of events of $w_1$ and $w_2$, annotated with the index of 
the word ($w_1$ or $w_2$) it belongs to. 
E.g. $\lpair{a}{0.4}\lpair{b}{2.1} \oplus \lpair{b}{0.3}\lpair{b}{0.4} = 
\lpair{(b,2)}{0.3}\lpair{(a,1)}{0.4}\lpair{(b,2)}{0.4}\lpair{(b,1)}{2.1}$.
The word $w_1 \oplus w_2$ is a timed word over the alphabet $\Sigma \times \{ 1,2\}$, 

\noindent{\bf Clocks.} Let $X$ be a set of clocks. 
A \emph{clock constraint} is a conjunction of terms of the form $x \op c$, 
where $x \in X$, $c \in \Ratp$ and $\op \in \{ <, \leq, =, \geq, > \}$. 
Let $B(X)$ denote the set of clock constraints. 
A \emph{clock valuation} $\nu$ is a mapping $\nu : X \mapsto \Realsp$.  

\noindent{\bf Timed automata.} A \emph{timed automaton} $\aut$ is a tuple $(\Sigma, L, l_0, X, 
\delta, F)$ where $\Sigma$ is the alphabet of $\cal A$, $L$ is a set of
locations, $l_0 \in L$ is the initial location, $X$ is a set of clocks, $\delta
\subseteq L \times \Sigma \times B(X) \times 2^X \times L$ is the switch
relation and ${F} \subseteq L$ is a set of accepting locations. 

\noindent{\bf Semantics of timed automata.} The semantics of a timed automaton $\aut$ is defined using an infinite-state
transition system $\ts{\aut}$ over the alphabet $(\Sigma \cup \{ \epsilon \})
\times \Realsp$.  A \emph{state} $q$ of $\ts{\aut}$ is a pair $(l,\nu)$
consisting of a location $l \in L$ and a clock valuation $\nu$. A state $q =
(l, \nu)$ satisfies a clock constraint $g$, denoted $q \models g$, if the
formula obtained from $g$ by substituting clocks from $X$ by their valuations
in $\nu$ is true.  There are two kinds of transitions in $\ts{\aut}$: $(i)$
\emph{elapse of time}: $(l, \nu) \transition{\tau} (l, \nu')$ iff for every $x
\in X$, $\nu'(x) = \nu(x) + \tau$ and $(ii)$ \emph{location switch}: $(l, \nu)
\transition{a} (l', \nu')$ iff there is a switch of ${\cal A}$, $(l, a, g,
\gamma, l')$, such that $(l, \nu) \models g$, and for each $x \in X$, 
$\nu'(x) = 0$ if $x \in \gamma$ and $\nu'(x) = \nu(x)$ otherwise. 
An elapse of time is usually followed by a
location switch. Thus we define the composition $\transition{\tau} \circ
\transition{a}$ and denote it as $\transitionComb{\tau}{a}$.
The initial state of $\ts{\aut}$ is the state $(l_0,\nu)$ where         
for each $x \in X$, $\nu(x) = 0$. The accepting states of $\ts{\aut}$
are all states of the form $\lpair{l}{\nu}$, where $l \in F$.
A \emph{run} of $\aut$ over a timed word $w = (a_0, t_0) (a_1, t_1) \ldots (a_k, t_k)$
is the sequence: $q_0 \transitionComb{t_0}{a_0} q_1 \transitionComb{t_1 -
t_0}{a_1} q_2 \ldots q_{k-1} \transitionComb{t_k - t_{k-1}}{a_k} q_{k+1}$, where $q_0$ is the initial state of $\ts{\aut}$.
The run is accepting if $q_{k+1}$ is an accepting state. The set of accepting
runs of $\aut$ is denoted $\traces{\aut}$. We say a timed word $w$ is accepted
by $\aut$ if there is a run in $\traces{\aut}$ whose projection to $\Sigma
\times \Realsp$ is $w$. 




The \emph{emptiness problem} for timed automata is as follows: given a timed
automaton $\cal A$, decide if $\traces{\cal A}$ is nonempty.  The emptiness
problem is also referred to as the \emph{reachability} problem as it is
equivalent to reachability of an accepting state in  $\ts{\aut}$.


\subsection{Weighted timed automata}

A \emph{weighted timed automaton} (\wta) is a timed automaton augmented by a function $C : L \cup \delta \mapsto \Rat$ that 
associates {\em weights} with the locations and switches of the timed automaton. 
The \emph{value} of a run $(l_0, \nu_0) \transition{\tau_0} (l_0, \nu_1) \transition{a_0} (l_1, \nu_2) \ldots \transition{a_{k}} (l_{k}, \nu_{2k+2})$ is given by 
$$ \sum_{i=0}^{k} C(l_{i}) \tau_{i} + \sum_{i=0}^k C(e_i)$$
where $e_i$ is the switch taken in the transition $(l_{i}, \nu_{2i+1}) \transition{a_i} (l_{i+1}, \nu_{2i+2})$. 
The value of a timed word $\Ts$ assigned by a \wta~$\aut$, denoted 
$\lang_\aut(\Ts)$, is defined as the infimum over values of all accepting runs of
$\aut$ on $\Ts$.  

The \emph{quantitative emptiness} problem for \wta~is as follows:
given a \wta~$\aut$ and $\lambda \in \Q$, decide if 
$\aut$ has an accepting run with value smaller than $\lambda$.

\begin{theorem}\cite{optimalPathsPspace}
The quantitative emptiness problem for \wta~is $\pspace$-complete.
\label{t:optTimedPaths}
\end{theorem}

A \wta~$\aut$ is \emph{functional} if for every timed word $\Ts$, all accepting runs of $\aut$ on $\Ts$ have the same value.

\subsection{Discrete transducers}

\newcommand{\run}{\gamma}

\noindent {\bf Discrete (finite-state) transducers.}\label{subsec:fst} A finite-state
transducer (\fst) $\tran$ is a tuple $(\alphI,\alphO,Q,Q_0,E,F)$ where
$\alphI$ is the input alphabet, $\alphO$ is the output alphabet, $Q$ is a
finite nonempty set of states, $Q_0 \subseteq Q$ is a set of initial states, $E
\subseteq Q \times \alphI \times \alphO^* \times Q$ is a set of
transitions, and $F$ is a set of accepting states.

\noindent {\bf Semantics of discrete transducers.}
A run $\run$ of $\tran$ on an input word $\s = \s[1]\s[2]\ldots \s[n]$ is defined in
terms of the sequence: $(q_0,\out_1)$, $(q_1,\out_2)$, $\ldots (q_1,\out_2)$ where $q_0 \in
Q_0$ and for each $i \in \{1,2,\ldots\}$, $(q_{i-1},\s[i],\out_i,q_{i}) \in E$. 
A run $(q_0,\out_1)$, $\ldots$ $(q_{n-1},\out_n)$, $(q_n,\es)$ 
is \emph{accepting} if $q_n \in F$. The output of $\tran$ along a run is the word 
$\out_1 \cdot \out_2 \cdot \ldots$ 
if the run is accepting, and is undefined otherwise. 
The transduction computed by an \fst $\tran$ 
is the relation $\funcDefinedBy{\tran} \subseteq \alphI^\omega 
\times \alphO^\omega$ (resp., $\funcDefinedBy{\tran} \subseteq \alphI^*  
\times \alphO^*$), where $(\s,\s') \in \funcDefinedBy{\tran}$ iff there is an
accepting run of $\tran$ on $\s$ with $\s'$ as the output along that run. 

\noindent {\bf Types of discrete transducers.} 
An \fst $\tran$ is called {\em functional} if the
relation $\funcDefinedBy{\tran}$ is a function. In this case, we use
$\funcDefinedBy{\tran}(\s)$ to denote the unique output word generated  along
any accepting run of $\tran$ on input word $\s$.  
An \fst is a letter-to-letter transducer if in every transition 
$(q,a,\out,a')$ we have $|\out| = 1$.

%

\section{Models of Timed I/O Systems}
\label{s:formalModels}
In this section, we present two models of timed I/O systems whose robustness will be studied in the following sections. 

\subsection{Timed transducers}

In this section, we define timed transducers, which extend 
classical discrete transducers.  


\begin{definition}[Timed transducer.] A \emph{timed transducer} $\trans$ is a timed automaton 
over an alphabet partitioned into an {\em input alphabet} $\alphI$ and an {\em output alphabet} 
$\alphO$.
\end{definition}

\noindent{\bf Semantics of timed transducers}. 
Given a timed transducer $\trans$, we define a relation $\funcDefinedBy{\trans} \subseteq \timedL{\alphI} \times \timedL{\alphO}$ 
by $\funcDefinedBy{\trans} = \{ (\Ts,\Tstr) : \trans$ accepts $\Ts \oplus   \Tstr \}$. 
We say that $\Tstr \in \timedL{\alphO}$
is an \emph{output} of $\trans$ on $\Ts \in \timedL{\alphI}$ if $(\Ts,\Tstr)  \in  \funcDefinedBy{\trans}$.

The following proposition we study the discrete parts of relations defined by timed transducers. 
We show that by imposing an additional assumption on transducers, namely that they do not have cycles labeled by $\alphO$,
we obtained the model  that  defines
relations on timed words such that their untimed parts can be defined by discrete transducers.
More formally, for a timed relation $R \subseteq \timedL{\alphI} \times \timedL{\alphO}$, 
we define $\untimed{R} \subseteq \alphI^* \times \alphO^*$ as
follows: for all $\s \in \alphI^*, \str \in \alphO^*$,
we have $(\s, \str) \in \untimed{R}$
 iff there exist $\Ts \in \timedL{\alphI}, \Tstr \in \timedL{\alphO}$ such that $(\Ts,\Tstr) \in R$, 
$\s = \untimed{\Ts}$ and $\str = \untimed{\Tstr}$. 

\begin{restatable}{proposition}{TimedExtensDiscrete}
\label{prop:TimedExtensDiscrete}
(i):~For every timed transducer $\trans$ that has no cycles labeled by $\alphO$, there exists a (nondeterministic) discrete transducer $\trans^d$
of exponential size in $|\trans|$ such that $\untimed{\funcDefinedBy{\trans}}$ and $\funcDefinedBy{\trans^d}$ 
coincide.
(ii):~For every discrete transducer $\trans^d$, there exists a timed transducer $\trans$ that has no cycles labeled by $\alphO$
such that $\untimed{\funcDefinedBy{\trans}}$ and $\funcDefinedBy{\trans^d}$ coincide.
\end{restatable}

\noindent{\bf Functionality}.
A transducer is \emph{timed-functional} iff $\funcDefinedBy{\trans}$ is a function, i.e.,
for all $\Ts \in \timedL{\alphI}$ and $\Tstr_1, \Tstr_2 \in \timedL{\alphO}$, if both
$(\Ts, \Tstr_1) \in \funcDefinedBy{\trans}$ and
$(\Ts, \Tstr_2) \in \funcDefinedBy{\trans}$, then $\Tstr_1 = \Tstr_2$.
For a timed-functional transducer $\trans$, we use $\funcDefinedBy{\trans}(\Ts)$ to denote 
the unique output of $\trans$ on $\Ts$. 

\begin{restatable}{proposition}{FunctionalityPspaceComplete}
Deciding timed functionality of a timed transducer is $\pspace$-complete.
\end{restatable}

Observe that a timed transducer does not have to be timed-functional, 
even if it is deterministic when viewed as a timed automaton.
Indeed, a trivial timed automaton that accepts every word over the alphabet $\alphI \cup \alphO$ 
is a deterministic and it is a timed transducer.
However, it is not functional.

In \ref{p:rigid-for-functional}, we present a sufficient condition for timed-functionality 
which can be checked in polynomial time. 
We further identify a class of transducers for which this condition is also necessary. 
A switch in a timed automaton is \emph{rigid} iff it is guarded by a constraint 
containing equality. A location $l$ in a timed automaton is unambiguous if
all constraints of any two outgoing switches from $l$ are strongly inconsistent, i.e.,
for all $x_1, \ldots, x_n, t$ the formula $g_1(x_1, \ldots, x_n) \wedge g_2(x_1 + t, \ldots, x_n+t)$ is does not hold.
A transducer is \emph{safe} if every location with outgoing $\alphI$ switches
is accepting.

\begin{restatable}{proposition}{SyntacticFunctionality}
(1)~A deterministic timed transducer in which all switches labeled by $\alphO$
are (a)~rigid, and (b)~all locations are with outgoing switches labeled with $\alphO$ are unambiguous,
 is functional.
(2)~Every function defined by a deterministic safe timed transducer is
also defined by a deterministic safe timed transducer satisfying (a) and (b) from (1).
\label{p:rigid-for-functional}
\end{restatable}

\subsection{Asynchronous Sequential Circuits}


\begin{figure}
\begin{center}
\genericAscDiagram
\caption{A generic \asc.}
\label{fig:asc}
\end{center}
\end{figure}

The second model of timed I/O systems that we consider is an asynchronous
sequential circuit (\asc). A generic \asc~is shown in \figref{asc} and some example 
\asc's are shown in Fig. \ref{f:osc}.

An \asc~is an I/O system composed of a combinational circuit (\coc)
and {\em memory devices}, or {\em delay elements}.  A \coc~is simply a
Boolean logic circuit that computes Boolean functions of its inputs.  A
\coc~is {\em memoryless}: the values of the circuit's output variables at time instant $t$
are functions of the values of the circuit's input variables at the same time instant
$t$.  A delay element is always labeled with some $d > 0$. The output of a
$d$-delay element at time $t$ equals its input at time $t-d$. 
We consider delays that are natural numbers.

\asc's may contain cycles, or {\em feedback loops}. Each such cycle is required
to contain at least one delay element. Due to the presence of delay elements
and feedback loops, an \asc~has memory: the outputs of an \asc~at time instant
$t$ are in general functions of its inputs at time instant $t$ as well as at
time instants $t'<t$.  The inputs of the delay elements of an \asc~are called 
{\em excitation variables}. The outputs of the delay elements of an \asc~are called 
{\em secondary variables}. The relationships between input, output, excitation and secondary variables 
of an \asc~are graphically represented in \figref{asc} and formally defined below.

\begin{definition}
Let $\ckt$ be an \asc~with input variables $\Inp = \ltuple{i}{m}$, output variables $\Out = \ltuple{o}{n}$, 
excitation variables $\Z = \ltuple{z}{k}$, secondary variables $\Y = \ltuple{y}{k}$ and delay 
elements $\Delta = \ltuple{d}{k}$.
Let $i(t)$ and $\Inp(t)$ denote the values of input $i$ and all inputs $\Inp$ at time $t$, respectively.
One can similarly define $o(t)$, $\Y(t)$  etc. We have the following:
\begin{align*}
&\forall j\in[1,k]: y^j(t) = \begin{cases}
	            0 &\mbox{if } t = [0,d^j)\\
                    z^j(t-d^j) &\mbox{if } t \geq d^j\\
		    \end{cases}\\
&\forall j\in[1,k]: z^j(t) = f^j(x^1(t),\ldots,x^m(t),y^1(t),\ldots,y^k(t)) \\
&\forall j\in[1,n]: o^j(t) = g^j(x^1(t),\ldots,x^m(t),y^1(t),\ldots,y^k(t)). \\
\end{align*}
Here, $f^1,\ldots,f^k$ and $g^1,\ldots,g^n$ are Boolean functions.
The {\em input alphabet} of \asc~$\ckt$, denoted $\alphI$, is given by $\{0,1\}^m$. 
The {\em output alphabet} of $\ckt$, denoted $\alphO$, is given by $\{0,1\}^n$. 
The \asc~$\ckt$ defines a transduction $\funcDefinedBy{\ckt} \subseteq \timedL{\alphI} \times
\timedL{\alphO}$ such that $\funcDefinedBy{\ckt}$ is a total function. Thus, the domain of 
$\ckt$ is given by $\dom{\ckt} = \timedL{\alphI}$.  We use $\funcDefinedBy{\ckt}(\Ts)$ to denote 
the unique output of $\ckt$ on $\Ts$.

\end{definition}



\section{Problem Statement}
\label{s:problem}
\noindent{\bf Similarity functions}. 
In our work, we use similarity functions to measure the similarity between timed words. 
Let $S$ be a set of timed words and let $\Realsinf$ denote the set $\Reals \cup \{\infty\}$.
A similarity function $d:S \times S\to \Realsinf$ is a function with the properties: 
$\forall x,y \in S:$ (1) $d(x,y) \ge 0$ and (2) $d(x,y) = d(y,x)$. 
A similarity function $d$ is also a distance (function or metric) if it satisfies the additional
properties: $\forall x,y,z \in S:$ (3) $d(x,y) = 0$ iff $x = y$ and (4)  $d(x,z)
\le d(x,y) + d(y,z)$. 
We emphasize that in our work we do not need to restrict similarity functions to be distances.

In this paper, we are interested in studying the {\em $K$-Lipschitz robustness} of
timed-functional transducers and \ascs. 

\begin{definition}[$K$-Lipschitz Robustness of Timed I/O Systems]
\label{def:robust}
Let $\tran$ be a timed-functional transducer or an \asc~with 
$\funcDefinedBy{\tran} \subseteq \timedL{\alphI} \times \timedL{\alphO}$.
Given a constant $K \in \Rat$ with $K > 0$ and similarity functions
$\dI: \timedL{\alphI} \times \timedL{\alphI} \; \to \; \Realsinf$
and $\dO: \timedL{\alphO} \times \timedL{\alphO} \; \to \; \Realsinf$,
the timed I/O system $\tran$ is called $K$-Lipschitz robust w.r.t. $\dI$, $\dO$ if:
\[\forall \Ts,\Tstr \in \dom{\tran}: \ \dI(\Ts,\Tstr) < \infty \, \Rightarrow \, 
\dO(\funcDefinedBy{\tran}(\Ts),\funcDefinedBy{\tran}(\Tstr)) \le K \dI(\Ts,\Tstr).
\]
\end{definition}

\section{Similarity Functions between Timed Words}
\label{s:formalDistances}
\noindent{\bf Timed words as Càdlàg functions}.  Consider a timed word $\Ts:
(a_0, t_0)(a_1, t_1) \ldots (a_k,t_k)$ over $(\Sigma, \timedom)$, where
$\timedom = [t_0, t_k]$ is an interval in $\Realsp$.  We define a Càdlàg
function $\ftime{\Ts}: \timedom \mapsto \Sigma$ as follows: for each $j \in
\{0,1,\ldots, k-1\}$, $\ftime{\Ts}(t) = a_j$ if $t \in [t_j,t_{j+1})$, and
$\ftime{\Ts}(t_k) = a_k$.  We define a timed word $\timed{\ftime{\Ts}} =
(\alpha_0,\tim_0)(\alpha_1,\tim_1)\ldots(\alpha_n,\tim_n)$ corresponding to the
Càdlàg function $\ftime{\Ts}$ such that:  for each $j \in \{0,1,\ldots,n\}$,
$\alpha_j = \ftime{\Ts}(\tim_j)$ and $\tim_j \in \{\tim_0,\ldots,\tim_n\}$ iff
$\ftime{\Ts}$ \emph{changes value} at $\tim_j$.  The timed word
$\timed{\ftime{\Ts}}$ can be interpreted as a {\em stuttering-free} version of
the timed word $\Ts$.


\noindent{\em Example}. Let $\Ts$ be the timed word
$(a,0)(b,1.3)(a,2)(a,2.9)(c,3.7)(a,5)$. Then $\ftime{\Ts}$ is given by the
following Càdlàg function over the interval $[0,5]$.

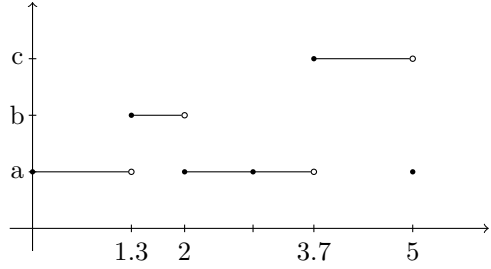
\begin{center}
\cadlagGraph
\end{center}
\noindent The timed word $\timed{\ftime{\Ts}}$ = $(a,0)(b,1.3)(a,2)(c,3.7)(a,5)$. 


\vspace{0.1in}

In what follows, let $\Ts$, $\Tstr$ be timed words over $(\Sigma, \timedom)$
with $\timedom \subseteq \mathbb{R}^+$. And let $\ftime{\Ts}$, $\ftime{\Tstr}$
be Càdlàg functions over $\timedom$ as defined above.  We present several
similarity functions between timed words  below. As will be clear, the
similarity between two timed words is computed as the similarity between their
corresponding Càdlàg functions.  We first present a similarity function between
discrete words. 

\noindent {\em Generalized Manhattan distance}.  The \emph{generalized
Manhattan distance} over discrete words $\s, \str$ is defined as: $\hd(\s,\str)
= \sum_{i=1}^{max(|\s|,|\str|)} \mathtt{diff}(\s[i],\str[i])$.  where
$\mathtt{diff}$ is the mismatch penalty for substituting letters.  The mismatch
penalty is required to be a distance metric on the alphabet (extended with a
special end-of-string letter $\eos$ for finite words).  When
$\mathtt{diff}(a,b)$ is defined to be $1$ for all $a,b$ with $a \neq b$, and
$0$ otherwise, $\hd$ is called the \emph{Manhattan distance}.



\vspace{0.1in}

\begin{definition}[Timed Manhattan distance]

Given $\mathtt{diff}$ on $\Sigma$:
\[
\dtM(\Ts,\Tstr) = \int_\timedom \mathtt{diff}(\ftime{\Ts}(x),\ftime{\Tstr}(x)) dx.
\]
\end{definition} 

Thus, the timed Manhattan distance extends the generalized Manhattan distance
by accumulating  the pointwise distance, as defined by $\mathtt{diff}$, between
the Càdlàg functions corresponding to timed words. 


\begin{definition}[Accumulated delay distance]
Let $\timed{\ftime{\Ts}} = (\alpha_0,\tim_0)(\alpha_1,\tim_1)\ldots(\alpha_n,\tim_n)$ 
and $\timed{\ftime{\Tstr}} = (\beta_0,\tau_0)(\beta_1,\tau_1)\ldots(\beta_n,\tau_m)$. 
\[
\dad(\Ts,\Tstr) = 
\begin{cases}
\sum_j |\tim_j - \tau_j| &\text{ if } \untimed{\timed{\ftime{\Ts}}} = \untimed{\timed{\ftime{\Tstr}}}\\
\infty &\text{ otherwise}.
\end{cases}
\]
\end{definition}

The accumulated delay distance examines the timed words $\timed{\ftime{\Ts}}$
and $\timed{\ftime{\Tstr}}$. If the projections of these timed words on their
$\Sigma$ components are equal, then the distance $\dad(\Ts,\Tstr)$ equals the
sum of delays between the corresponding events; otherwise the distance is
infinite. 


\begin{definition}[Skorokhod distance w.r.t. timed Manhattan distance]
Let $\Wiggles$ be the set of all continuous bijections from the domain $I$ of 
$\ftime{\Ts}$ and $\ftime{\Tstr}$ onto itself.
\[
\dS(\ftime{\Ts},\ftime{\Tstr}) = 
\inf_{\wiggle \in \Wiggles} \left(\lone{\mathtt{Id} - \wiggle} + \dtM(\ftime{\Ts},\ftime{\Tstr} \circ \wiggle)\right),
\]
where $\mathtt{Id}$ is the identity function over $I$, $\lone{.}$ is the $L_1$-norm over $\Realsp$ 
and $\circ$ is the usual function composition operator. 
\end{definition}

The Skorokhod distance is a popular distance metric for continuous functions.
Hence, it is also a natural choice for our Càdlàg functions.  The Skorokhod
distance permits {\em wiggling} of the function values as well as the timeline 
in order to {\em match up} the functions.
The timeline wiggle is executed using continuous bijective functions, denoted 
$\wiggle$, over the timeline. 
The first component of the Skorokhod distance measures the magnitude of the 
{\em timing distortion} resulting from a timeline wiggle $\wiggle$. 
The second component of the Skorokhod distance measures the magnitude of the 
{\em function value mismatch} under $\wiggle$. 
The Skorokhod distance is the least value obtained over all such timeline wiggles. 
The magnitudes of the timing distortion and function value mismatch can be computed 
and combined in different ways. In our work, the timing distortion is computed 
as the $L_1$ norm, the function value mismatch is computed as the timed Manhattan distance 
and the two are combined using addition. 

We now present some helpful connections between the above distances.  

\begin{restatable}{proposition}{RelationBetweenDistances}[Relations between distances]
(i)~The accumulated delay distance coincides with the Skorokhod distance w.r.t. the timed Manhattan distance
defined by $\mathtt{diff}^{=}$ such that: 
$\forall a,b \in \alphI$, $\mathtt{diff}^{=}(a,b) = 0$ if $a = b$ and $\mathtt{diff}^{=}(a,b) = \infty$ otherwise. 
(ii)~For every timed Manhattan distance $\dtM^{\leq 1}$ defined 
$\mathtt{diff}^{\leq 1}$ such that $\forall a,b \in \alphI$, $\mathtt{diff}^{\leq 1}(a,b) \leq 1$, 
we have the Skorokhod distance w.r.t. $\dtM^{\leq 1}$ coincides with $\dtM^{\leq 1}$. 
\label{RelationBetweenDistances}
\end{restatable}


\section{Robustness Analysis of Timed Transducers}
\label{s:robustnessAutomata}

\noindent{\bf Timed-automatic similarity function.} 
A timed similarity function $\dist$ is \emph{computed} by a \wta~$\aut$ iff for all $\Ts, \Tstr \in \timedL{\Sigma}$,
$\dist(\Ts,\Tstr) = \lang_{\aut}(\Ts \oplus \Tstr)$.
A timed similarity function $\dist$ computed by a \wta~is called a \emph{timed-automatic similarity function}.

\noindent{\bf $N$-interleaved timed words.} 
Timed words $\Ts, \Tstr$ are defined to be {\em $N$-interleaved} iff 
in any time interval $[t_1,t_2]$, the numbers of events from $\Ts$ and from $\Tstr$ differ by at most $N$.
Intuitively, the $N$-interleaved property 
expresses that two words are synchronized~\cite{DBLP:conf/hybrid/HenzingerO14}.

\begin{definition}
A timed-functional transducer $\trans$ is called \emph{timed-synchronized} iff there exists $N$ 
such that for every $\Ts \in \timedL{\alphI}$, the words $\Ts$ and $\funcDefinedBy{\trans}(\Ts)$ are 
$N$-interleaved.
\end{definition}

\begin{restatable}{theorem}{AutomataBasedRobustness}
\label{t:automata}
Let $\dI$, $\dO$ be timed-automatic similarity functions such that
$\dI, \dO$ are computed by (nondeterministic) \wta. 
\begin{compactenum}[(i)]
\item There exists a sound procedure for checking 
$K$-robustness of a timed-synchronized transducer w.r.t. $\dI, \dO$ that works in polynomial space. 
\item If $\dO$ is computed by a functional \wta, checking 
$K$-robustness of a timed-synchronized transducer w.r.t. $\dI, \dO$ is \pspace-complete.
\end{compactenum}
\end{restatable}




In what follows, we define several timed similarity functions that can be
computed by functional and nondeterministic \wta.\\

\noindent {\bf Timed similarity functions computed by functional \wta}.
We show that the timed Manhattan and accumulated delay distances 
can be computed by functional \wta. 

\begin{restatable}{lemmaStatement}{TimedManhattanIsAutomatic}
\label{l:Man}
The timed Manhattan distance $\dtM$ over timed words 
is computed by a functional \wta.
\end{restatable}

To compute the timed Manhattan distance, the \wta~simply tracks the
$\mathtt{diff}$ between timed events using its weight function. The semantics
of \wta~then imply that the value assigned by the automaton
to a pair of timed words is precisely the timed Manhattan distance between
them. 

%

\begin{restatable}{lemmaStatement}{AccumulatedDelayIsAutomatic}
\label{l:Ad}
Let $\dur$, $\bound$ be any nonnegative real numbers. 
The accumulated delay distance $\dad$ over  
timed words $\Ts$, $\Tstr$ such that:
\begin{compactenum}
\item the duration of any segment in $\ftime{\Ts}$, $\ftime{\Tstr}$ is greater than $\dur$
and 
\item the delay $|\delta_j - \tau_i|$ between corresponding events in 
$\ftime{\Ts}$, $\ftime{\Tstr}$ is less than $\bound$, 
\end{compactenum}
is computed by a functional \wta.
\end{restatable}

The \wta~tracks with its weight function the number of unmatched events.
Again, the semantics of \wta~imply that the value assigned by the automaton
to a pair of timed words is precisely the accumulated delay distance.
To make sure that every event is matched to the right event, i.e. the untimed parts are equal, the automaton
implements a buffer to store the unmatched events.
The assumptions on the minimal duration of events and the maximal delay between the corresponding events 
imply that the buffer's size is bounded.\\

\noindent {\bf Timed similarity functions computed by nondeterministic \wta}.
A (restricted) Skorokhod distance can be computed by a nondeterministic \wta. 
We first prove the following lemma characterizing an essential 
subset of the set $\Wiggles$ of all timing distortions. 

\begin{restatable}{lemmaStatement}{SkorohodIsPiecewiseLinear}[Skorokhod distance is realized by a piecewise linear function]
\label{l:piece}
Let $\Ts$, $\Tstr$ be timed words. 
Let $\eta$ be the number of segments in $\Tstr$. 
For every $\epsilon > 0$, there exists a piecewise linear function $\wiggle$ consisting of 
$\eta$ segments such that 
$\lone{\mathtt{Id} - \wiggle} + \dtM(\ftime{\Ts},\ftime{\Tstr} \circ \wiggle)
- \dS(\ftime{\Ts},\ftime{\Tstr})| \leq \epsilon$.
\end{restatable}

Lemma~\ref{l:piece} implies that $\lone{\mathtt{Id}} - \wiggle$ coincides with the accumulated delay distance between
$\ftime{\Tstr}$ and $\ftime{\Tstr} \circ \wiggle$. This allows us to compute the Skorokhod distance by a \wta~for
$\wiggle$ for which there is a \wta~that can compute the accumulated delay between 
$\ftime{\Tstr}$ and $\ftime{\Tstr} \circ \wiggle$. 

\begin{restatable}{lemmaStatement}{SkorohodIsAutomatic}
\label{l:S}
Let $\dur$, $\bound$ be any nonnegative real numbers. 
The Skorokhod distance $\dS$ over  
timed words $\Ts$, $\Tstr$ restricted to time distortions $\wiggle$ such that:
\begin{compactenum}
\item the duration of any segment in $\ftime{\Tstr}$, $\ftime{\Tstr} \circ \wiggle$ is greater than $\dur$
and 
\item the delay $|\delta_j - \tau_i|$ between corresponding events in 
$\ftime{\Tstr}$, $\ftime{\Tstr} \circ \wiggle$ is less than $\bound$, 
\end{compactenum}
is computed by a nondeterministic \wta.
\end{restatable}


\begin{remark}
Physical systems typically 
have a bounded rate at which they can generate/process data. 
Hence, bounding the minimum possible duration of timed symbols 
is not a severe restriction from the modeling perspective.
Moreover, if an input is delayed arbitrarily, it 
makes little sense to constraint the system behavior. 
Hence, for robustness analysis, it is also reasonable to 
bound the maximum delay between corresponding events.  
\end{remark}

\noindent{\bf Summary of decidability results}. We summarize the decidability results 
for timed-synchronized transducers  
that follow from Theorem \ref{t:automata} and Lemmas \ref{l:Man}, \ref{l:Ad} and \ref{l:S}.  
\begin{compactenum}
\item $K$-robustness  
is \pspace-complete for timed Manhattan distances.
\item $K$-robustness  
is \pspace-complete for accumulated delay distances,  
under environment assumptions from Lemma \ref{l:Ad}.
\item $K$-robustness  
is \pspace-complete if the input perturbation is computed 
as a Skorokhod distance and the output perturbation is computed 
as a timed Manhattan distance. 
\end{compactenum}

\section{Robustness Analysis of Asynchronous Sequential Circuits}
\label{s:robustnessCircuits}
\setlength{\textfloatsep}{10pt}

In this section we show that robustness of \ascs~w.r.t. the timed Manhattan distances 
is \pspace-complete. The decision procedure is by reduction to discrete
letter-to-letter transducers. Our argument consists of two steps and relies 
on the use of steps functions --- Càdlàg functions that change values only at integer points.
First, we show that on inputs that are step functions,
$\ascs$~behave like discrete letter-to-letter transducers. Second, we show that 
if an \asc~is not $K$-robust w.r.t. the timed Manhattan distances, there exists 
a counterexample consisting of a pair of inputs that are step functions. 






\noindent{\bf \ascs~transforming step functions.}
There is a natural correspondence between step functions $f : [0, T] \mapsto \{0,1\}^k$
and words over the alphabet $\{0,1\}^k$. The function $f$ defines the word $w_f = f(0) f(1) \ldots f(T-1)$ and,
conversely, a word $w \in (\{0,1\}^k)^*$ defines a step function $f_w$ such that  
$w_{f_w} = w$. We aim to show that the behavior of \ascs~on step function $f$ is captured 
by discrete transducers on words $w_f$. 

First, observe that an \asc~with integer delays transforms step functions into step functions. 
Indeed, the output at time $t$ depends on the input and secondary variables at time $t$, which are 
equal to the values of  excitation variables at times $\ltuple{t-d}{k}$. The excitation variables at times
 $\ltuple{t-d}{k}$ depend on inputs and secondary variables at times $\ltuple{t-d}{k}$. As delays are integers, 
by unraveling the definition of the output variables (resp., excitation and secondary variables) at time $t$, we
obtain that they depend solely on (a subset of) inputs at times $frac(t), frac(t)+1, \ldots, t$, where $frac(t)$ is the fractional part of $t$.
Therefore, if an input is a step function, then excitation, secondary and output variables are all step functions. 
Moreover, the value of the step function output in the interval $[j, j+1)$ with $j \in \N$ can be computed using 
the input value in the interval $[j, j+1)$ and the values of excitation variables in the intervals $[j-d^1, j+1-d^1), \ldots [j-d^k, j+1-d^k)$.
Therefore, we can define a discrete letter-to-letter transducer that simulates the given \asc. Such a transducer 
remembers in its states values of the excitation variables in the last $\max(d^1, \ldots, d^k)$ 
intervals.

\begin{restatable}{lemmaStatement}{AscsAreDiscrete}
(1)~If the input to an \asc~is a step function, the output is a step function.
(2)~Given an \asc~$\ckt$, one can compute in polynomial space a discrete letter-to-letter transducer $\trans_{\ckt}$ such that
for every step function $f$, the output of $\ckt$ on $f$ is $f_v$, where $v$ is the output of $\trans$
on $w_f$.
\label{l:AscsAreDiscrete}
\end{restatable}
 
\begin{remark}
The transducer $\trans_{\ckt}$ in Lemma~\ref{l:AscsAreDiscrete} can be constructed in polynomial space, meaning that
its sets of states and accepting states are succinctly representable and we can decide in polynomial time whether a given 
tuple $(q,a,b,q')$ belongs to the transition relation of  $\trans_{\ckt}$.
\end{remark}

\noindent{\bf Counterexamples to $K$-robustness of~\ascs.}
Consider an \asc~with integer delays that is not $K$-robust w.r.t. $\dI, \dO$. 
Then, there are two input functions $f_1, f_2$ that witness non-$K$-robustness, i.e.,
$\dO(\funcDefinedBy{\ckt}(f_1),\funcDefinedBy{\ckt}(f_2)) > K \cdot \dI(f_1, f_2)$.
We show that for \ascs, if there exists a pair of functions that witnesses non-$K$-robustness, 
there exists a pair of step functions that witnesses non-$K$-robustness as well.
Recall that the output of the \asc~at time $t$ depends only on inputs at times $frac(t), frac(t)+1, \ldots, t$.
Hence, we argue that if the pair $f_1, f_2$ is a witness of non-$K$-robustness, then for some $x \in [0,1)$,
$f_1, f_2$ restricted to the domain  $\Delta_x = \{ y \in \dom{f_1} \cap \dom{f_2} \mid frac(y) = x \}$ is also a witness of non-$K$-robustness. 
Since the set $\Delta_x$ is discrete, we can define step functions based on $f_1, f_2$
restricted to $\Delta_x$.

\begin{restatable}{lemmaStatement}{stepFunctionsCounterexaples}
Let $\ckt$ be an \asc~with integer delay elements.
If $\ckt$ is not $K$-robust w.r.t. timed Manhattan distances $\dI, \dO$, then
there exists a pair of step functions $f_1, f_2$ such that 
$\dO(\funcDefinedBy{\ckt}(f_1),\funcDefinedBy{\ckt}(f_2)) > K \cdot \dI(f_1, f_2)$.
\label{l:stepFunctionsCounterexaples}
\end{restatable}

\noindent{\bf $K$-robustness of discrete transducers}.
We next present a decidability result that follows from \cite{DBLP:journals/corr/HenzingerOS14}.
Deciding $K$-robustness of letter-to-letter transducers w.r.t. generalized
Manhattan distances reduces to quantitative non-emptiness of weighted automata
with $\fsum$-value function~\cite{DBLP:journals/corr/HenzingerOS14}.  The
latter problem can be solved in nondeterministic logarithmic space, assuming
that the weights are represented by numbers of logarithmic length. Hence, we obtain the 
following result for {\em short} generalized Manhattan distances, i.e., distances 
whose $\mathtt{diff}$  values are represented by numbers of logarithmic length. 


\begin{lemma}
Deciding $K$-robustness of letter-to-letter transducers w.r.t. short generalized Manhattan distances 
is in \nlogspace.
\label{l:discretedDecidablity}
\end{lemma}

We can now characterize the complexity of checking $K$-robustness of \ascs. 

\begin{restatable}{theorem}{AscRobustnessPspace}
Deciding $K$-robustness of \ascs~with respect to timed Manhattan distances is \pspace-complete.
\end{restatable}

\begin{figure}
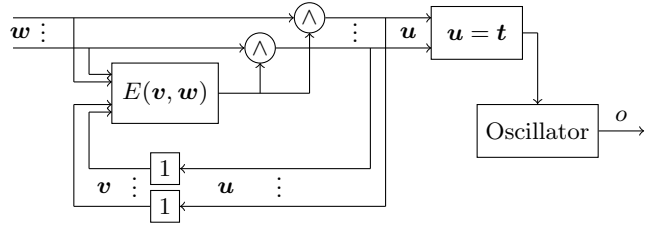

\begin{center}
\succintReachabilityCircuit
\end{center}
\caption{The diagram of an~\asc~from the reduction of the reachability in succinctly represented graphs to $K$-robustness of~\ascs.}
\label{fig:ASCreduction}
\end{figure}

\begin{proof}
Observe that the timed Manhattan distance between step functions $f,g$ equals the generalized Manhattan distance between
the words $w_f, w_g$ corresponding to step functions $f,g$. This, together with 
Lemmas \ref{l:AscsAreDiscrete} and \ref{l:stepFunctionsCounterexaples}, allows us to reduce checking $K$-robustness 
of \ascs~w.r.t. timed Manhattan distances to checking $K$-robustness of the corresponding letter-to-letter transducers w.r.t. generalized Manhattan distances.  
It then follows from Lemma~\ref{l:discretedDecidablity} that checking $K$-robustness of \ascs~is in \pspace.
Note that generalized Manhattan distances are short in this case as their descriptions are logarithmic 
in the exponential size of the letter-to-letter transducer.

The \pspace-hardness of checking $K$-robustness of \ascs~is obtained by a reduction from 
the reachability problem for succinctly represented graphs, which is \pspace-complete~\cite{lozano1990complexity}.
Succinctly represented graphs are given indirectly by a propositional formula $E(\vec{w}, \vec{v})$, where $\vec{w}, \vec{v}$ are vectors of $n$ variables.
The vertexes of the graph are binary sequences of length $n$, and two sequences are connected by an edge iff the formula $E(\vec{w}, \vec{v})$ 
on these sequences holds.
Consider the graph $G$ represented by the formula $E(\vec{v},\vec{w})$ and its vertex $\vec{t}$. 
We claim that the~\asc~given in \figref{ASCreduction} is $K$-robust iff the vertex $\vec{t}$ is 
not reachable from the zero vector $(0,\ldots, 0)$ in $G$. 
Due to Lemma~\ref{l:stepFunctionsCounterexaples} it suffices to focus on inputs that are step functions $f$, or discrete words $w_f$.
The input is interpreted as a sequence of vertexes of $G$. 
The \asc~in \figref{ASCreduction} consists of (a) a circuit $E(\vec{v},\vec{w})$ which checks whether there is an edge between
$\vec{v}$ and the input $\vec{w}$, (b) a unit that tests whether $\vec{u}$ equals the target vertex $\vec{t}$ and, 
(c) an oscillator (\ref{ex:oscillator}) which outputs $0$ when the input is $0$, and once the input is $1$, outputs $1$ 
until the end of the input. 
Initially, $\vec{v}$ is the zero vector. If there is an edge between $\vec{v}$ and $\vec{w}$, 
$\vec{u}$ is set to $\vec{w}$, and hence, $\vec{v}$ equals $\vec{w}$ in the next step and $\vec{w}$ is checked for equality with $\vec{t}$. 
If $\vec{w} = \vec{t}$, the oscillator is activated.
Otherwise, if there is no edge between $\vec{v}$ and $\vec{w}$, $\vec{u}$ is set to the zero vector, which corresponds to transitioning back to the initial vertex;
$\vec{v}$ equals the zero vector in the next step and the zero vector is checked for equality with $\vec{t}$.


If $\vec{t}$ is not reachable from the zero vector, the output of the~\asc~is always $0$, and hence 
the \asc~is $K$-robust for every $K$. 
Conversely, we claim that if $t$ is reachable from the zero vector, then the~\asc~is not $K$-robust for any $K$.
Indeed, consider a shortest path from the zero vector to the target vertex $\vec{0}, \vec{v}_1, \ldots, \vec{t}$ and 
consider the following two inputs: $i_1 = \vec{0}, \vec{v}_1, \ldots, \vec{t}, \vec{0}^K$, the path leading to activation
of the oscillator followed by $K$ inputs that are zero vectors, and, 
$i_2 = \vec{0}, \vec{v}_1, \ldots, \vec{t'}, \vec{0}^K$, which is obtained from $i_1$ by changing one bit in $\vec{t}$.
Observe that the oscillator in~\asc~is not activated on the input $i_2$, hence the output is $0$.
Therefore, while the timed Manhattan distance between the inputs is $1$,
the timed Manhattan distance between the outputs is $K+1$, for any chosen $K$. 

\end{proof}


\begin{remark}
Recall that the domain of an \asc~$\ckt$ with input alphabet $\alphI = \{0,1\}^m$ 
is given by $\dom{\ckt} = \timedL{\alphI}$. 
For any timed Manhattan distance $\dtM{\leq 1}$ over $\dom{\ckt}$ 
such that $\forall a,b \in \alphI$, $\mathtt{diff}^{\leq 1}(a,b) \leq 1$, 
Proposition \ref{RelationBetweenDistances} states
that the Skorohod distance w.r.t. $\dtM^{\leq 1}$ coincides with $\dtM^{\leq 1}$. 
Hence, $K$-robustness w.r.t. such Skorokhod distances is \pspace-complete as well. 
\end{remark}


\bibliography{papers}

\appendix
\section{Proofs from Section~\ref{s:formalModels}}

\TimedExtensDiscrete*
\begin{proof}
\noindent\textbf{(i):} Consider a timed transducer $\trans$. Let $\aut$ be a finite automaton 
that accepts the untimed language of $\trans$. Such an automaton exists, it is 
of exponential time and can be constructed out of the region graph for $\trans$~\cite{alur1994theory}.
Since $\trans$ does not have cycles labeled by $\alphO$, all path in $\aut$ labeled by $\alphO$ are finite.
We build a transducer $\trans^d$ from $\aut$ in the following way:
$\trans^d$ has the same set of states, the same initial state and the same accepting states as $\aut$.
Next, $(q,a,w,q')$ is a transition of $\trans^d$ iff there is a path in $\aut$ labeled with
$aw$, where $a \in \alphI$, $w \in \alphO$ and that path cannot be extended by a transition labeled with $\alphO$.
It follows from construction $\untimed{\funcDefinedBy{\trans}}$ and $\funcDefinedBy{\trans^d}$ coincide. 

\noindent\textbf{(ii):} Consider a discrete transducer $\trans^d$. We construct 
a timed I/O automaton $\trans$ without clocks from $\trans^d$.
Basically, we substitute each transition $(q,a,w,q')$ with $w = w[1] \ldots w[k]$
by a path $(q, a, T, \emptyset, q^{(q',w,0)})$,
$(q^{(q',w,0)}, a, T, \emptyset, q^{(q',w,1)}), \ldots, (q^{(q',w,k)}, a, T, \emptyset, q')$.
Clearly, $\untimed{\funcDefinedBy{\trans}}$ and $\funcDefinedBy{\trans^d}$ coincide.
\end{proof}




\FunctionalityPspaceComplete*
\begin{proof}
\textbf{Containment in $\pspace$:} We construct a timed automaton $\aut$ as $\trans \times \trans' \times \aut^{\neq}$, where
\begin{enumerate}
\item $\trans$ is a transducer from $\timedL{\alphI}$ to $\timedL{\alphO}$,
\item $\trans'$ is a transducer from $\timedL{\alphI}$ to $\timedL{\alphO'}$, where $\alphO' = \{ a' : a \in \alphO\}$ is disjoint from $\alphI, \alphO$,
\item $\aut^{\neq}$ is an automaton that works over $\timedL{\alphO \cup \alphO'}$ and 
accepts languages of words $\Ts \otimes \Tstr$, with $\Ts \in \timedL{\alphO}$, $\Tstr \in \timedL{\alphO'}$, such that
$\Ts$ does not correspond to $\Tstr$, i.e., if every event $(a', t)$ in $\Tstr$ is translated to $(a,t)$,
the resulting timed word is different than $\Ts$.
\end{enumerate}
The timed automaton $\aut$ accepts timed words over the alphabet $\alphI \cup \alphO \cup \alphO'$ that are counterexamples to 
the functionality property of $\trans$. Therefore, functionality of $\trans$ reduces in polynomial time to the emptiness problem 
for timed automata.
 
\textbf{$\pspace$-hardness:} We reduce the emptiness problem 
for timed automata to deciding functionality of timed transducer.
Given a timed automaton $\aut$, we transform it into a timed transducer $\trans^{b}$ by
substituting each switch $s = (l, a, g, X, l')$ with two switches 
$(l, a, g, X \cup {x^g}, l^g)$ and $(l^g, a, g \wedge g = 0, X, l')$, where $l^g$ is a fresh location and $l^g$ is a new clock common
for all new switches. Basically, the transducer $\trans$ implements the identity function on the timed language of $\aut$.
Now, we construct a transducer $\trans'$ from $\timedL{\alphI}$ to $\timedL{\alphI \cup \{\bot\}}$, such that
for every $(a_0, t_0) \ldots (a_n, t_n) \in \timedL{\alphI}$,
$(\bot, t_0) \ldots (\bot, t_n) \in \funcDefinedBy{\trans'}$. Then, such a transducer $\trans'$ is functional
iff $\aut$ accepts the empty language. 
\end{proof}

\SyntacticFunctionality*
\begin{proof}
\noindent\textbf{(1):} Consider a two timed words $\Ts, \Tstr$ over the alphabet $\alphI \cup \alphO$ such that
their projections on events over $\alphI$ are equal. We can prove by induction on the number of events in $\Ts, \Tstr$
that $\Ts, \Tstr$ are equal. 
Assume that $\Ts$ and $\Tstr$ are equal up to event $i$. Therefore, $\trans$ is in the same state  $(l, \nu)$
upon reading first $i$ events of $\Ts$ as $\Tstr$.
If $\Ts[i+1]$ and $\Tstr[i+1]$ are both input events, then it is the same event by the assumption on projections. 
Otherwise, $l$ has an outgoing switch labeled by $\alphO$, hence it is an unambiguous location.
It follows that there is exactly one switch $s$ outgoing of $l$ whose guard is satisfied by $\nu + t$ for some $t$.
In consequence, the untimed parts of $\Ts[i+1]$ and $\Tstr[i+1]$ are equal. 
Moreover, the guard of $s$ contains equality, therefore the time that $\trans$ spends in $l$ is uniquely determined, i.e.,
the timestamps of $\Ts$ and $\Tstr$ are equal. It follows that $\Ts[i+1] = \Tstr[i+1]$.
 
\noindent\textbf{(2):} Consider a deterministic functional timed transducer $\trans$ and its switch 
$(l, a, g, X, l')$. 
First, we claim that in every accepting run, every time the switch $(l, a, g, X, l')$
is taken, the value of at least one clock is equal to a constant from $g$, therefore it can be replaced
by at most linearly many (in the size of $\trans$) rigid switches.

Towards contradiction; suppose that there is an accepting run $\pi$
in which the switch is taken at a position $i$ at which no clock value equals to any guard from $g$. 
Then, consider two runs $\pi_1$, obtained by truncating $\pi$ to the first positions after $i$ at which
a switch labeled with $\alphI$ is taken (or just $\pi$ if there is no such position).
Next, $\pi_2$ obtained from $\pi_1$ by increasing the time spend in $l$ by small time so that the guard of $g$
are still satisfied. Since $\trans$ is a safe transducer,  both runs $\pi_1$ and $\pi_2$ are accepting.
Observe that both runs projected on 
events from $\alphI$ are the same. However, projections of $\pi_1$ and $\pi_2$ on events from $\alphO$
are different, which contradicts functionality of $\trans$.

Second, consider a location $l$ with outgoing switches labeled by $\alphO$.
Observe that either $l$ does not have outgoing switches labeled by $\alphI$ 
or there is no accepting run going through $l$ that takes a switch labeled with $\alphO$.
Indeed, if there is such a run and $l$ has an outgoing switch labeled with $\alphI$, then
$l$ is an accepting location. Hence, the run truncated to the position of $l$ is accepting 
and the run truncated to the first accepting position past $l$ is also accepting. 
Projections of those two runs on $\alphO$ are different, but projections on $\alphI$ 
are equal, which contradicts functionality.
It follows that we may transform the transducer to an equivalent one, whose
locations have either all outgoing switches labeled with $\alphI$ or 
all outgoing switches labeled with $\alphO$.

Finally, we can extend guards of each switch by the full information about the timed automata region,  i.e.,
a switch $s$ with a guard $g$ is substituted with switches $s_1, \ldots, s_k$ with guards $g_1, \ldots, g_k$
that are maximal conjunctions of inequalities of clocks that are consistent with $g$. 
Next, we remove switches that are not taken in any accepting run. 

Consider a location $l$ with all outgoing switches labeled by $\alphO$.
We claim that (*)~we can enrich each switch guarded by $g$ by formulas $\forall t. \neg g'(\vec{x} + t)$, where $g'$
are guards of other switches, does not change the set of accepted runs.
Observe that linear arithmetic admits quantifier elimination, hence $\forall t. \neg g'(\vec{x} + t)$ can be change to a quantifier free formula, which can be written in 
a disjunctive normal form $g'_1 \vee \ldots \vee g'_p$ such conjunctions $g_i'$ and $g_j'$ are inconsistent for $i \neq j$.
Then, we substitute the switch with the guard $g$ by $p$ switches with the same locations, label and reset variables, but 
guards $g \wedge g_1', \ldots, g \wedge g_p'$. Such guards are strongly inconsistent.

It remains to prove $(*)$. If an accepting run contains a state $(l,\nu)$ such that $\nu$ satisfies the guard $g$ of a switch $s$ and for some $t \geq 0$, 
$\nu + t$ satisfies the guard $g'$ of a switch $s'$ with $s \neq s'$, then the transducer is not functional. 
Indeed, we consider two cases. Assume that $t = 0$. Since the transducer is deterministic, $s$ and $s'$ are labeled with different letter from $\alphO$.
Similarly to the previous cases, we can construct two accepting runs, that are identical till the location $l$, but then one takes $s$
and terminates as soon as it reaches accepting location. The other run takes $s'$ ans also terminates as soon as it reaches an accepting location. 
Such two runs violate functionality property.
Assume that $t > 0$. Again we can construct two runs that violate functionality as the time spent in $l$ is different for both runs.  

\end{proof}

\begin{remark}
The safety assumption in (2)~of~Proposition~\ref{p:rigid-for-functional} is essential. 
Indeed, consider a function $f$ defined on the domain $\{ (a,0), (a, x+1) : x \in \R^+\}$ as 
$f((a, 0) (a, x+1)) = (b, x) (b, x+1)$. This function can be represented by 
a deterministic functional transducer that accepts words $(a,0), (b,x), (a,x+1) (b,x+1)$.
The switch taken on the event $(b,x)$ cannot be rigid though. 
Intuitively, in deterministic timed transducers, the timestamp of every output has to 
be fixed w.r.t. to the input events. But, it can be fixed w.r.t. proceeding input events,
or, as in the case of $f$, it can be fixed w.r.t. some of the proceeding output event. 
Unfortunately, this implies that necessary and sufficient conditions for
functionality of deterministic transducers are non-local and involve some
reachability-base conditions, which are usually $\pspace$-hard for timed automata.
\end{remark}

%
%

\section{Proofs from Section~\ref{s:robustnessAutomata}}

For our automata constructions, we find it helpful to view timed words as
starting with symbol $\sos$ and ending with symbol $\eos$.  Given alphabet
$\alphI$, let $\alphI^\eos$ denote $\alphI \cup \{\eos\}$ and $\seos{\alphI}$
denote $\alphI \cup \{\sos,\eos\}$. For a string $w$, we use $w[i]$ to refer to
the $i^{th}$ letter of $w$, with the first letter at index $0$.

\AutomataBasedRobustness*
\begin{proof}
Given automata $\autI$ computing $\dI$,
$\autI$ computing  $\dO$ and the timed transducer $\trans$ we construct a weighted timed automaton $\aut$
such that $\aut$ has a run of value less than $0$ iff $\trans$ is not $K$-robust.
We first define variants of $\autI$, $\autO$ and $\trans$ to enable these 
automata to operate over a common alphabet $\Lambda = \alphI \oplus \alphI
\oplus \alphO \oplus \alphO$.
In particular, we define automata
$\bautI$, $\bautO, \blaut_{\tran}, \braut_{\tran}$ on $\Lambda$
such that for all $\Ts,\Tstr,\Ts',\Tstr'$:
\begin{compactenum}
\item the value of $\bautI$ on $\Ts \oplus \Tstr \oplus \Ts' \oplus \Tstr'$
is equal to the value of $\autO$ on $\Ts \oplus \Tstr$,
\item the value of $\baut_{\dO}$ on $\Ts \oplus \Tstr \oplus \Ts' \oplus \Tstr'$
is equal to the value of $\aut_{\dO}$ on $\Ts' \oplus \Tstr'$,
\item  $\blaut_{\tran}$ accepts $\Ts \oplus \Tstr \oplus \Ts' \oplus \Tstr'$
iff $\aut_{\tran}$ accepts $\Ts \oplus \Ts'$, and
\item $\braut_{\tran}$ accepts $\Ts \oplus \Tstr \oplus \Ts' \oplus \Tstr'$
iff $\aut_{\tran}$ accepts $\Tstr \oplus \Tstr'$.
\end{compactenum}

Let $\bautI^{K}$, $\bautO^{-1}$ be weighted timed automata obtained by
multiplying each transition weight of $\bautI$, $\bautO$ by $K$,
$-1$, respectively.
Consider the \twa~$\aut$ defined as $\bautI^K \times
\blaut_{\tran} \times \braut_{\tran} \times \bautO^{-1}$, the synchronized product
of automata $\bautI^K, \blaut_{\tran}, \braut_{\tran}, \bautO^{-1}$ where
the weight of each transition is equal the sum of the weights of the corresponding transitions
in $\bautI^K$ and $\bautO^{-1}$.

Now, we show that there exists a word on with value below $0$ assigned by $\aut$ 
iff $\trans$ is not $K$-robust w.r.t. $\dI, \dO$.

%

Consider words $\Ts,\Tstr, \Ts',\Tstr'$ such that 
$\aut$ accepts $\Ts \oplus \Tstr \oplus \Ts' \oplus \Tstr'$.
The value assigned by $\aut$ to this timed word 
equals $K\dI(\Ts,\Tstr) + \inf_{\pi \in \Acc } -val_\pi$,
where $\Acc$ is the set of accepting runs of $\autO$ 
on $\Ts' \otimes \Tstr'$, $val_\pi$ denotes the value of run $\pi$. 
 
If the language of $\aut$ is empty for threshold $0$, it means 
that for all words $\Ts$, $\Tstr$, 
$K\dI(\Ts,\Tstr) \geq \inf_{\pi \in \Acc } val_\pi$. 
By the definition of $\autO$, $\inf_{\pi \in \Acc } val_\pi = 
\dO(\Ts',\Tstr')$. Hence, it follows that $\tran$ is $K$-robust. 

If $\autO$ is functional,
each run in $\Acc$ has the same value $\dO(\Ts', \Tstr')$.
Thus, $\lang_{\aut}(\Ts \oplus \Tstr \oplus \Ts' \oplus \Tstr')$ 
equals $K\dI(\Ts,\Tstr) - \dO(\Ts', \Tstr')$, and,
$\lang_{\aut}(\Ts \oplus \Tstr \oplus \Ts' \oplus \Tstr') < 0$ implies $\tran$ is \emph{not} $K$-robust.
Conversely, if $\tran$ is not $K$-robust there are words $\Ts, \Tstr$ such that
$\dO(\funcDefinedBy{\trans}(\Ts), \funcDefinedBy{\trans}(\Tstr)) >
K\dI(\Ts,\Tstr)$.  This implies $\aut$ accepts 
$\Ts \oplus \Tstr \oplus 
\funcDefinedBy{\trans}(\Ts) \oplus \funcDefinedBy{\trans}(\Tstr)$ and
$\lang_{\aut}(\Ts \oplus \Tstr \oplus \funcDefinedBy{\trans}(\Ts) \oplus
\funcDefinedBy{\trans}(\Tstr)) < 0$.
Thus, nonemptiness of $\aut$ and $K$-robustness of $\tran$ w.r.t. $\dI, \dO$ coincide.

\end{proof}

\TimedManhattanIsAutomatic*
\begin{proof}

Let $\aut = (\alphI \times \{1,2\}, L, \ell_0, X, F, \delta,C)$ 
be a \twa~where:
\begin{compactitem}
\item $L = \{\ell_{a,b}: a,b \in \seos{\alphI}\}$
\item $\ell_0 = \ell_{\sos,\sos}$
\item $X = \{x\}$
\item $(\ell_{a,b}, \symb, g, \reset, \ell'_{a',b'}) \in \delta$ iff $g = \true$, $\reset = \{ \}$, and exactly one of the following holds:
\begin{compactenum}
\item $\symb = (c,1)$ with $c \in \alphI$, $a' = c$ and $b' = b$, or, 
\item $\symb = (c,2)$ with $c \in \alphI$, $a' = a$ and $b' = c$ 
\end{compactenum}
\item $F = \ell_{\eos,\eos}$
\item For each $\ell_{a,b} \in L$: $C(\ell_{a,b}) = \diff(a,b)$ and for each $e \in \delta$: $C(e) = 0$
\end{compactitem}

Observe that $\aut$ is deterministic and $\lang_\aut(\s \oplus \str) = \dtM(\s,\str)$.
\end{proof}

\AccumulatedDelayIsAutomatic*
\begin{proof}

Let $M = \lceil \bound/\dur\rceil$.   
In the following, we assume that the symbol duration of any timed symbol is greater than or equal to $\dur$, 
the delay between corresponding events is less than or equal to $\bound$, 
and that timed words are well-formed, i.e., there is no timed symbol after the $\eos$ symbol  
(we do not check for any of these). 

Let $\dot{\alphI} = \{\dot{a}| a \in \alphI\}$. 

Let $\aut = ((\alphI \cup \dot{\alphI}) \times \{1,2\}, L, \ell_0, X, F, \delta,C)$ 
be a \twa~where:
\begin{compactitem}
\item $L = \{\ell_{(w,i)}: i \in \{1,2\}, |w| \leq M+1, w \in (\eps \cup \dot{\alphI}).\alphI^*\} \cup 
           \{\ell_{(w.\eos,i)}: i \in \{1,2\}, |w| \leq M, w \in \alphI^*\} \cup 
           \{\ell_{\sos}, \ell_{\eos}, \ell_{\rej}\}$
\item $\ell_0 = \ell_{\sos}$
\item $X = \{x\}$
\item $(\ell_{\alpha}, \symb, g, \reset, \ell'_{\alpha'}) \in \delta$ iff $g = \true$, $\reset = \{ \}$ and one of the following holds:
\begin{compactenum}
\item For $i \in\{1,2\}$: $\alpha = \sos$, $\symb = (c,i)$ with $c \in \alphI^\eos$,  and $\alpha' = (c,i)$, or, 
\item For $i \in\{1,2\}$: $\alpha = (w,i)$ with $w < M$, $\symb = (c,i)$ with $c \in \alphI$ and $c \neq w[|w|-1]$,  and $\alpha' = (w.c,i)$, or, 
\item For $i \in\{1,2\}$: $\alpha = (w,i)$ with $w < M$, $\symb = (c,i)$ with $c \in \alphI$ and $c = w[|w|-1]$,  and $\alpha' = (w,i)$, or, 
\item For $i \in\{1,2\}$: $\alpha = (w,i)$ with $w = M$, $\symb = (c,i)$ with $c \in \alphI$, and $\alpha' = \rej$, or, 
\item For $i \in\{1,2\}$: $\alpha = (w,i)$ with $w \leq M$, $\symb = (\eos,i)$, and $\alpha' = (w.\eos,i)$, or, 
\item For $i \in\{1,2\}$: $\alpha = (w,i)$ with $|w| > 1$, $w = c.x$, $\symb = (c,j)$ with $j = 3-i$, and $\alpha' = (\dot{c}.x,i)$
\item For $i \in\{1,2\}$: $\alpha = (\dot{f}.w,i)$ with $|w| > 1$, $w = c.x$, $\symb = (c,j)$ with $j = 3-i$, and $\alpha' = (x,i)$
\item For $i \in\{1,2\}$: $\alpha = (\dot{f}.w,i)$, $w = x$, $\symb = (c,j)$ with $c = f$, $j = 3-i$, and $\alpha' = (\dot{f}.w,i)$
\item For $i \in\{1,2\}$: $\alpha = (\dot{f}.w,i)$ with $|w| > 1$, $w = c.x$, $\symb = (d,j)$ with $d \neq c$ and $d \neq f$, 
$j = 3-i$, and $\alpha' = \rej$
\item For $i \in\{1,2\}$: $\alpha = (\dot{f}.w,i)$ with $w < M$, $\symb = (c,i)$ with $c \in \alphI$ and $c \neq w[|w|-1]$,  and $\alpha' = (\dot{f}.w.c,i)$, or, 
\item For $i \in\{1,2\}$: $\alpha = (\dot{f}.w,i)$ with $w < M$, $\symb = (c,i)$ with $c \in \alphI$ and $c = w[|w|-1]$,  and $\alpha' = (\dot{f}.w,i)$, or, 
\item For $i \in\{1,2\}$: $\alpha = (c,i)$, $\symb = (c,j)$ with $c \in \alphI$, $j = 3-i$, and $\alpha' = \sos$ 
\item For $i \in\{1,2\}$: $\alpha = (\eos,i)$, $\symb = (\eos,j)$, $j = 3-i$, and $\alpha' = \eos$ 
\item For $i \in\{1,2\}$: $\alpha = (d,i)$ with $d \neq c$, $\symb = (c,j)$ with $c \in \alphI^\eos$, $j = 3-i$, and $\alpha' = \rej$ 
\item $\alpha = \rej$, $\symb = *$, and $\alpha' = \rej$
\end{compactenum}
\item $F = \ell_{\eos}$
\item For each $\ell_{w} \in L$: $C(\ell_{w}) = |w|$ and for each $e \in \delta$: $C(e) = 0$
\end{compactitem}

Observe that $\aut$ is deterministic. We claim that $\lang_\aut(\Ts \oplus
\Tstr) = \dad(\Ts,\Tstr)$.  The main insight is as follows. $\dad(\Ts,\Tstr)$
is the sum of the {\em waiting times} for every symbol of
$\timed{\ftime{\Ts}}$, $\timed{\ftime{\Tstr}}$ for its matching symbol from
$\timed{\ftime{\Tstr}}$, $\timed{\ftime{\Ts}}$, respectively. State
$\ell_{(w,1)}$ stores the subword $w$ of $\untimed{\timed{\ftime{\Ts}}}$ that has arrived
already, and is waiting to be matched with the corresponding subword of
$\untimed{\timed{\ftime{\Tstr}}}$.  Thus, as long as a symbol $c$ of $w$ is not consumed by a
matching symbol of $\untimed{\timed{\ftime{\Tstr}}}$, we need to count the duration spent
waiting for $(c,2)$. This equals the sum of the time spent in each state
$\ell_{(x,1)}$,  visited since seeing $(c,1)$ until seeing $(c,2)$. The above
cost function ensures that the value of a run on $\Ts \oplus \Tstr$ equals
$\dad(\Ts,\Tstr)$. Once a symbol $c$ of $w$ is consumed by a matching symbol 
of $\untimed{\timed{\ftime{\Tstr}}}$, one needs to disregard subsequent 
$c$ symbols of $\Tstr$ without trying to match them 
with symbols of $\Ts$. This is because we are tracking the distance between 
$\timed{\ftime{\Ts}}$ and $\timed{\ftime{\Tstr}}$, and not $\Ts$ and $\Tstr$. 
This is taken care of using states of the form $(\dot{f}.w,i)$, which remember 
the symbol $f$ to be disregarded. 
\end{proof}

\SkorohodIsPiecewiseLinear*
\begin{proof}
\newcommand{\tilda}[1]{\widetilde{#1}}
Let $v$ be a timed word and the domain of $\ftime{v}$ is $[a,b]$.
Consider two continuous bijections from $[a,b]$ onto itself, $\wiggle_1, \wiggle_2$.
Observe that if $\wiggle_1, \wiggle_2$ agree on timestamps of the events of $v$, i.e.,
for every event $(a,t) \in u$ we have $\wiggle_1(t) = \wiggle_2(t)$ then 
$\ftime{v} \circ \wiggle_1 = \ftime{v} \circ \wiggle_2$. 

Now, let $\tilda{\wiggle} \in \Wiggles$ satisfy 
$|\lone{I - \wiggle} + \dtM(\ftime{u},\ftime{v} \circ \tilda{\wiggle}) - \dS (\ftime{u},\ftime{v})| \leq \epsilon$.
Consider a piecewise linear function $\wiggle'$ consisting of $|v|$ segments that
agrees with $\tilda{\wiggle}$ on the timestamps of the events of $v$. 
Then, $\ftime{v} \circ \tilda{\wiggle} = \ftime{v} \circ \wiggle'$ and 
$|\lone{I - \wiggle} + \dtM(\ftime{u},\ftime{v} \circ \wiggle') - \dS (\ftime{u},\ftime{v})| \leq \epsilon$.
\end{proof}

\SkorohodIsAutomatic*
\begin{proof}
Consider an alphabet $\alphI \times \{1,2,3\}$. We consider words over such an alphabet to be the disjoint union of three
words $\Ts_1, \Ts_2$ and $\Ts_3$ denoted by $\Ts_1 \oplus \Ts_2 \oplus \Ts_3$.
First, we construct a weighted timed automaton $\aut_1$, which on a word $\Ts_1 \oplus \Ts_2 \oplus \Ts_3$ 
computes the sum of the timed Manhattan distance between words $\Ts_1$ and $\Ts_3$ and
the lossy accumulated delay distance between $\Ts_2$ and $\Ts_3$. 
The automaton $\aut_1$ is a product of weighted timed automata 
that compute the timed Manhattan distance and the loosy accumulated delay. 
The automaton $\aut_2$ is a projection of $\aut_1$ on  $\alphI \times \{1,2\}$, i.e.,
it computes $\inf_{\Ts_3}  \dad(\Ts_2, \Ts_3) + \dtM(\Ts_1,\Ts_3)$. 
Observe that $\Ts_3$ can be considered as $\Ts_2 \circ \lambda$ and 
$\dad(\Ts_2,\Ts_3)$ coincides with the $L_1$-norm of $I -\lambda$.
\end{proof}

\section{Proofs from Section~\ref{s:robustnessCircuits}}

\AscsAreDiscrete
\begin{proof}
\noindent\textbf{(1):} It readily follows from  the discussion above Lemma~\ref{l:AscsAreDiscrete}.

\noindent\textbf{(2):} 
Let $M$ be the maximal delay in a given \asc.
The discrete transducer $\trans$ stores the sequence of excitation variables from the last $M+1$ rounds
$\vec{z}_0, \ldots, \vec{z}_M$, i.e., the state space
is $(\{0,1\}^k)^M$. At each step, $\trans$ shifts stored excitation variables and computes
the new value of the most recent excitation variables $\vec{z}_0$
and the output variables $\vec{o}$  using Boolean function $f,g$. In these functions, 
the values of secondary variables are obtained from appropriately delayed excitation variables.

Observe that the size of $\trans$ is exponential in the number of variables. However, 
the set of states has compact representation $(\{0,1\}^k)^M$, as well as the input and output alphabets 
$\{0,1\}^m$ and respectively $\{0,1\}^n$. Moreover, given Boolean vectors $q,q',a,b,$ of lengths $kM,kM,m,n$
$q, a,b,q'$, we can compute in polynomial time whether $\trans$ has a transition from $q$ to $q'$ upon reading $a$
at which it outputs $q'$.
Finally, $\trans$ is a deterministic letter-to-letter transducer.
\end{proof}

\stepFunctionsCounterexaples*
\begin{proof}
Consider functions $f_1, f_2$ on the domain $[0, T]$ that witness non-$K$-robustness of a given \asc, i.e.,
$\dO(\funcDefinedBy{\ckt}(f_1),\funcDefinedBy{\ckt}(f_2)) - K \cdot \dI(f_1, f_2) > 0$.
Recall that the output of the \asc~at time $t$ depends only on inputs at times $frac(t), frac(t)+1, \ldots, t$.
Therefore, we can consider separately $f_1, f_2$ and their corresponding outputs at times from $T_x  = \{ x + i : i \in \N,  x+ i \leq T \}$, i.e., 
reals from $[0,T]$ with the fractional part $x$. The value of 
$\dO(\funcDefinedBy{\ckt}(f_1),\funcDefinedBy{\ckt}(f_2)) - K \cdot \dI(f_1, f_2)$ on $T_x$ is a finite sum
$\sum_{t \in T_x} (\diffO(\funcDefinedBy{\ckt}(f_1)(t),\funcDefinedBy{\ckt}(f_2)(t)) - K \diffI(f_1(t), f_2(t)))$.
We observe that the value of $\dO(\funcDefinedBy{\ckt}(f_1),\funcDefinedBy{\ckt}(f_2)) - K \cdot \dI(f_1, f_2)$  on $[0,T]$
is the integral over $[0,1)$ of $\sum_{t \in T_x} (\diffO(\funcDefinedBy{\ckt}(f_1)(t),\funcDefinedBy{\ckt}(f_2)(t)) - K \diffI(f_1(t), f_2(t)))$ 
considered as a function of $x$. It follows that if the given $\asc$ is not $K$-robust
then there exists $x \in [0,1)$ such that $\sum_{t \in T_x} (\diffO(\funcDefinedBy{\ckt}(f_1)(t),\funcDefinedBy{\ckt}(f_2)(t)) - K \diffI(f_1(t), f_2(t)))$
is strictly positive. Clearly, step functions $g_1, g_2$ defined on each interval
$[i, i+1)$ to be equal to $f_1(x+i)$ and respectively $f_2(x+i)$ satisfy 
$\dO(\funcDefinedBy{\ckt}(g_1),\funcDefinedBy{\ckt}(g_2)) - K \cdot \dI(g_1, g_2) >0$, i.e., they witness a non-$K$-robustness of the given \asc. 
\end{proof}


\end{document}